%% file: petr4-parsers.tex
\documentclass[sigplan,acm,10pt,screen]{acmart}

\setcopyright{rightsretained}
\acmPrice{}
\acmDOI{10.1145/3519939.3523715}
\acmYear{2022}
\copyrightyear{2022}
\acmSubmissionID{pldi22main-p419-p}
\acmISBN{978-1-4503-9265-5/22/06}
\acmConference[PLDI '22]{Proceedings of the 43rd ACM SIGPLAN International Conference on Programming Language Design and Implementation}{June 13--17, 2022}{San Diego, CA, USA}
\acmBooktitle{Proceedings of the 43rd ACM SIGPLAN International Conference on Programming Language Design and Implementation (PLDI '22), June 13--17, 2022, San Diego, CA, USA}

\usepackage{etoolbox}
\apptocmd{\sloppy}{\hbadness 10000\relax}{}{} % chktex 1

\usepackage{booktabs}   %% For formal tables:
\usepackage[skip=6pt,belowskip=-10pt]{caption}

\usepackage[belowskip=2pt]{subcaption}

\makeatletter
\def\@acmplainindent{0pt}
\def\@acmdefinitionindent{0pt}
\def\@proofindent{\noindent}
\makeatother

\usepackage{xparse}
\usepackage{multirow}
\usepackage{mathpartir}
\usepackage[capitalize,nameinlink,noabbrev]{cleveref}
\crefname{algocf}{Algorithm}{Algorithms}
\crefname{examplecf}{Example}{Examples}
\crefname{section}{\S}{\S}
\Crefname{section}{Section}{Sections}
\crefname{subsection}{\S}{\S}
\Crefname{subsection}{Section}{Sections}
\crefname{equation}{}{}
\Crefname{equation}{Equation}{Equations}
\Crefname{lemma}{Lemma}{Lemmas}
\Crefname{theorem}{Theorem}{Theorems}
\Crefname{proposition}{Proposition}{Propositions}
\Crefname{corollary}{Corollary}{Corollaries}
\Crefname{algocf}{Algorithm}{Algorithms}
\Crefname{appsec}{Appendix}{Appendices}
\crefname{appsec}{\S}{\S}

\newcommand{\LTac}{$\text{L}_{\text{tac}}$}

\newcommand{\STAB}[1]{\begin{tabular}{@{}c@{}}#1\end{tabular}}

\newcommand{\angl}[1]{\langle#1\rangle}
\newcommand{\naturals}{\mathbb{N}}

\newcommand{\bits}{{\{0,1\}}}

\newcommand{\sfunc}{\mathsf{sz}}

\newcommand{\acc}{\mathsf{accept}}
\newcommand{\rej}{\mathsf{reject}}

\usepackage{stmaryrd}

\newcommand{\sem}[1]{\mathrel{\left\llbracket#1\right\rrbracket}}
\newcommand{\psem}[3]{\mathrel{\left\llbracket#1\right\rrbracket_{#2}^{#3}}}
\newcommand{\lsem}[2]{\sem{#1}^{#2}_{\mathcal{L}}}
\newcommand{\bsem}[3]{\ensuremath{{\left\llbracket#1\right\rrbracket}_{\mathcal{B}}^{#2}(#3)}}

\newcommand{\semexp}[1]{\ensuremath{\isem{#1}_{\mathcal{E}}}}
\newcommand{\semop}[1]{\ensuremath{\isem{#1}_{\mathcal{O}}}}
\newcommand{\sempat}[1]{\ensuremath{\isem{#1}_{\mathcal{P}}}}
\newcommand{\semtrans}[1]{\ensuremath{\isem{#1}_{\mathcal{T}}}}
\newcommand{\semaut}[1]{\isem{#1}_{\mathcal{A}}}
\newcommand{\opsize}[1]{|\!|#1|\!|}

\newcommand{\typ}[2]{\ensuremath{\vdash_{#1} #2}}
\newcommand{\typbare}[1]{\ensuremath{\vdash_{#1}}}
\newcommand{\typexp}[2]{\ensuremath{\typ{\mathcal{E}}{#1 : #2}}}
\newcommand{\typexpbare}{\ensuremath{\typbare{\mathcal{E}}}}
\newcommand{\typop}[1]{\ensuremath{\typ{\mathcal{O}}{#1}}}
\newcommand{\typopbare}{\ensuremath{\typbare{\mathcal{O}}}}

\newcommand{\typtrans}[1]{\ensuremath{\typ{\mathcal{T}}{#1}}}
\newcommand{\typtransbare}{\ensuremath{\typbare{\mathcal{T}}}}
\newcommand{\typaut}[1]{\ensuremath{\typ{\mathcal{A}}{#1}}}
\newcommand{\typautbare}{\ensuremath{\typbare{\mathcal{A}}}}

\newcommand{\Var}{\mathsf{Var}}

\usepackage[ruled,lined,linesnumbered,noend]{algorithm2e}

\newcommand{\WP}{\mathsf{WP}}

\newcommand{\floor}[1]{\lfloor#1\rfloor}
\newcommand{\reach}{\mathsf{reach}}

\SetKwInput{Input}{Input}
\SetKwInput{Output}{Output}
\SetKw{True}{true}
\SetKw{False}{false}
\SetKw{KwNot}{not}
\SetKw{Continue}{continue}

\newcommand{\Smany}[1]{\overline{#1}}
\newcommand{\Sst}{\ensuremath{\mathit{st}}\xspace}

\newcommand{\Sop}{\ensuremath{\mathit{op}}}

\newcommand{\Sseq}[2]{\ensuremath{#1;\,#2}}
\newcommand{\Sextract}[1]{\ensuremath{\operatorname{\mathtt{extract}}(#1)}}
\newcommand{\Sassign}[2]{\ensuremath{#1 \mathrel{:=} #2}}
\newcommand{\Sexp}{\ensuremath{e}}
\newcommand{\Shdr}{\ensuremath{h}}
\newcommand{\Sbits}{\ensuremath{\mathit{bv}}}
\newcommand{\Snat}{\ensuremath{n}}
\newcommand{\Sslice}[3]{\ensuremath{#1[#2\!:\!#3]}}
\newcommand{\Sconcat}[2]{\ensuremath{#1 \mathrel{++} #2}}
\newcommand{\Strans}{\ensuremath{\mathit{tz}}}
\newcommand{\Sgoto}[1]{\ensuremath{\operatorname{\mathtt{goto}}(#1)}}
\newcommand{\Ssel}[2]{\ensuremath{\operatorname{\mathtt{select}}(#1) \{#2\}}}
\newcommand{\Scase}{\ensuremath{\mathit{c}}}
\newcommand{\Spat}{\ensuremath{\mathit{pat}}}
\newcommand{\Sany}{\ensuremath{\_}}
\newcommand{\Sstate}[3]{#1 \, \{ #2; #3 \}}
\newcommand{\Saut}{\ensuremath{\mathit{aut}}}
\newcommand{\Sbexp}{\ensuremath{\mathit{be}}}
\newcommand{\Sbbuf}[1]{\ensuremath{\mathrm{buf}^{#1}}}
\newcommand{\Sbhdr}[1]{\ensuremath{h^{#1}}}
\newcommand{\Svar}{\ensuremath{x}}
\newcommand{\Spred}{\ensuremath{p}}
\newcommand{\Sstatepred}[2]{\ensuremath{{#1}^{#2}}}
\newcommand{\Sbuflen}[2]{\ensuremath{{#1}^{#2}}}
\renewcommand{\S}{\ensuremath{p}}

\newcommand{\syndef}[3]{#1&::=&#2&\text{#3} \\}
\newcommand{\syncase}[2]{&|&#1&\text{#2} \\}

\newcommand{\isem}[1]{\llbracket#1\rrbracket}

\newcommand{\SYSTEM}{Leap\-frog\xspace}
\usepackage[nameinlink]{cleveref}

\usepackage{thm-restate}
\usepackage{float}
\usepackage{paralist}

\usepackage{xstring}
\usepackage{catchfile}
\IfFileExists{../.git/HEAD}{%
    
    \StrGobbleRight{\input{../.git/HEAD}}{1}[\head]
    \StrBehind[2]{\head}{/}[\branch]
    \IfFileExists{../.git/refs/heads/\branch}{%
        
    }{%
        
    }
}

\usepackage{multicol}

\usepackage{array}
\usepackage{arydshln}
\usepackage{listings}

\definecolor{dkgreen}{rgb}{0,0.6,0}
\definecolor{dkblue}{rgb}{0,0,0.6}
\definecolor{ltblue}{rgb}{0,0.4,0.4}
\definecolor{dkviolet}{rgb}{0.3,0,0.5}
\definecolor{mygray}{RGB}{245,245,245}

\newcommand{\code}[1]{\texttt{#1}}

\lstdefinelanguage{Coq}{
    mathescape=true,
    texcl=false,
    morekeywords=[1]{Section, Module, End, Require, Import, Export,
        Variable, Variables, Parameter, Parameters, Axiom, Hypothesis,
        Hypotheses, Notation, Local, Tactic, Reserved, Scope, Open, Close,
        Bind, Delimit, Definition, Let, Ltac, Fixpoint, CoFixpoint, Add,
        Morphism, Relation, Implicit, Arguments, Unset, Contextual,
        Strict, Prenex, Implicits, Inductive, CoInductive, Record,
        Structure, Canonical, Coercion, Context, Class, Global, Instance,
        Program, Infix, Theorem, Lemma, Corollary, Proposition, Fact,
        Remark, Example, Proof, Goal, Save, Qed, Defined, Hint, Resolve,
        Rewrite, View, Search, Show, Print, Printing, All, Eval, Check,
        Projections, inside, outside, Def},
    morekeywords=[2]{forall, exists, exists2, fun, fix, cofix, struct,
        match, with, end, as, in, return, let, if, is, then, else, for, of,
        nosimpl, when},
    morekeywords=[3]{Type, Prop, Set, true, false, option},
    morekeywords=[4]{pose, set, move, case, elim, apply, clear, hnf,
        intro, intros, generalize, rename, pattern, after, destruct,
        induction, using, refine, inversion, injection, rewrite, congr,
        unlock, compute, ring, field, fourier, replace, fold, unfold,
        change, cutrewrite, simpl, have, suff, wlog, suffices, without,
        loss, nat_norm, assert, cut, trivial, revert, bool_congr, nat_congr,
        symmetry, transitivity, auto, split, left, right, autorewrite},
    morekeywords=[5]{by, done, exact, reflexivity, tauto, romega, omega,
        assumption, solve, contradiction, discriminate},
    morekeywords=[6]{do, last, first, try, idtac, repeat},
    morecomment=[s]{(*}{*)}, % chktex 9
    showstringspaces=false,
    morestring=[b]", % chktex 18
    morestring=[d],
    tabsize=3,
    extendedchars=false,
    sensitive=true,
    breaklines=false,
    basicstyle=\small,
    captionpos=b,
    columns=[l]flexible,
    identifierstyle={\ttfamily\color{black}},
    keywordstyle=[1]{\ttfamily\color{dkviolet}},
    keywordstyle=[2]{\ttfamily\color{dkgreen}},
    keywordstyle=[3]{\ttfamily\color{ltblue}},
    keywordstyle=[4]{\ttfamily\color{dkblue}},
    keywordstyle=[5]{\ttfamily\color{dkred}},
    stringstyle=\ttfamily,
    commentstyle={\ttfamily\color{dkgreen}},
    literate=
    {\\forall}{{\color{dkgreen}{$\forall\;$}}}1
    {\\exists}{{$\exists\;$}}1
    {<-}{{$\leftarrow\;$}}1
    {=>}{{$\Rightarrow\;$}}1
    {==}{{\code{==}\;}}1
    {==>}{{\code{==>}\;}}1
    {->}{{$\rightarrow\;$}}1
    {<->}{{$\leftrightarrow\;$}}1
    {<==}{{$\leq\;$}}1
    {\#}{{$^\star$}}1
    {\\o}{{$\circ\;$}}1
    {\@}{{$\cdot$}}1
    {\/\\}{{$\wedge\;$}}1 % chktex 4
    {\\\/}{{$\vee\;$}}1 % chktex 4
    {++}{{\code{++}\;}}1
    {~}{{\ }}1
    {\@\@}{{$@$}}1
    {\\mapsto}{{$\mapsto\;$}}1
    {\\hline}{{\rule{\linewidth}{0.5pt}}}1 % chktex 44
}[keywords,comments,strings]

\usepackage[framemethod=TikZ]{mdframed}

\mdfsetup{
  roundcorner=5pt,
  splittopskip=0,%
  splitbottomskip=0,%
  frametitleaboveskip=0,
  frametitlebelowskip=0,
  skipabove=0,%
  skipbelow=0,%
  leftmargin=0,%
  rightmargin=0,%
  innertopmargin=0mm,%
  innerbottommargin=0mm,%
}
\lstnewenvironment{lf-yellow}{\mdframed}{\endmdframed}
\lstnewenvironment{lf-grey}{\mdframed[linecolor=black!25,backgroundcolor=black!5]}{\endmdframed}

\begin{document}

%% Title information
\title{\SYSTEM:\texorpdfstring{\@}{} Certified Equivalence for Protocol Parsers}

%% Author information
%% Contents and number of authors suppressed with 'anonymous'.
%% Each author should be introduced by \author, followed by
%% \authornote (optional), \orcid (optional), \affiliation, and
%% \email.
%% An author may have multiple affiliations and/or emails; repeat the
%% appropriate command.
%% Many elements are not rendered, but should be provided for metadata
%% extraction tools.

%% Author with single affiliation.
\author{Ryan Doenges}
\affiliation{
  \institution{Cornell University}
  \country{USA}
}
\email{rhd89@cornell.edu}

\author{Tobias Kapp\'{e} }
\affiliation{
  \institution{ILLC, University of Amsterdam}
  \country{Netherlands}
}
\email{t.kappe@uva.nl}

\author{John Sarracino}
\affiliation{
  \institution{Cornell University}
  \country{USA}
}
\email{jsarracino@cornell.edu}

\author{Nate Foster}
\affiliation{
  \institution{Cornell University}
  \country{USA}
}
\email{jnfoster@cs.cornell.edu}

\author{Greg Morrisett}
\affiliation{
  \institution{Cornell University}
  \country{USA}
}
\email{jgm19@cornell.edu}

%% Abstract
%% Note: \begin{abstract}...\end{abstract} environment must come
%% before \maketitle command
\begin{abstract}
\input{abstract}
\end{abstract}

%% 2012 ACM Computing Classification System (CSS) concepts
%% Generate at 'http://dl.acm.org/ccs/ccs.cfm'.
\begin{CCSXML}
  <ccs2012>
     <concept>
         <concept_id>10003752.10003766.10003773</concept_id>
         <concept_desc>Theory of computation~Automata extensions</concept_desc>
         <concept_significance>500</concept_significance>
         </concept>
     <concept>
         <concept_id>10011007.10010940.10010992.10010998.10010999</concept_id>
         <concept_desc>Software and its engineering~Software verification</concept_desc>
         <concept_significance>500</concept_significance>
         </concept>
   </ccs2012>
\end{CCSXML}

%% End of generated code

%% Keywords
%% comma separated list
%\keywords{keyword1, keyword2, keyword3}  %% \keywords are mandatory in final camera-ready submission

%% \maketitle
%% Note: \maketitle command must come after title commands, author
%% commands, abstract environment, Computing Classification System
%% environment and commands, and keywords command.
\maketitle

\input{intro}
\input{overview}
\input{automata}
\input{equivalence}
\input{optimizing}

\input{implementation}
\input{case-studies}
\input{related}

\input{acks}

%% Bibliography
\balance
\bibliography{bibliography}

%% Appendix
% when submitting, comment out the following

\clearpage
\appendix
\input{appendix}

\end{document}

%% file: abstract.tex
We present \SYSTEM, a Coq-based framework for verifying
equivalence of network protocol parsers. Our approach is based on an
automata model of P4 parsers, and an algorithm for symbolically computing
a compact representation of a bisimulation, using ``leaps.''
Proofs are powered by a certified compilation chain from first-order entailments
to low-level bitvector verification conditions,
which are discharged using off-the-shelf SMT solvers.
As a result, parser equivalence proofs in \SYSTEM are fully automatic and push-button.

We mechanically prove the core metatheory that underpins our approach,
including the key transformations and several optimizations.
We evaluate \SYSTEM on a range of practical
case studies, all of which require minimal configuration and no manual
proof. Our largest case study uses \SYSTEM to perform translation
validation for a third-party compiler from automata to hardware
pipelines. Overall, \SYSTEM represents a step towards a world where
all parsers for critical network infrastructure are verified. It
also suggests directions for follow-on efforts, such as verifying
relational properties involving security.

%% file: intro.tex
\section{Introduction}
Devices like routers, firewalls and network interface cards as well as
operating system kernels occupy a critical role in modern communications
infrastructure.
Each of these implements parsing for a cornucopia of networking
protocols in its \emph{protocol parser}.
The parser is the network's first line of defense, responsible for organizing
and filtering unstructured and often untrusted data as it arrives from the
outside world.
Due to their crucial role, bugs in parsers are a significant source of
crashes, vulnerabilities, and other faults~\cite{langsec}.

\paragraph{Example Router Bug.} Consider the following bug,
which was present in a commercial router developed by a leading
equipment vendor several years ago.
Internally, the router was organized around a high-throughput
pipeline, which most packets traversed in a single pass.
However some packets had to be \emph{recirculated}, meaning they took additional
passes through the pipeline before being sent back out on the wire.
The router used an internal state variable to decide whether a packet should be
recirculated.
Usually this state variable was initialized by vendor-supplied code.
But, as was discovered by a customer, it could also be erroneously
initialized from data in non-standard, malformed packets.
Hence, crafted packets could bypass the vendor-supplied initialization code,
resulting in an infinite recirculation loop---a denial-of-service (DoS) attack
on the router and its peers.
In the presence of broadcast traffic, such a ``packet storm'' would monopolize
the router's resources, rendering it unusable until it was rebooted.

An easy way to avoid this bug would be to modify the router's parser
to filter away malformed packets, while still accepting valid packets.
However, to have full confidence in the new parser, one would need to
prove that it is \emph{equivalent} to the original, modulo malformed
packets.
Although parsers tend to be simple, this would likely be a challenging
verification task---it requires reasoning about a \emph{relational} property
across two distinct programs.

\paragraph{Parser Equivalence Checking.}
This paper studies relational verification of protocol parsers,
focusing specifically on equivalence of parsers expressed in terms of
state machines.
Semantic equivalence~\cite{10.1145/3314221.3314596} is a fundamental problem that underpins a wide range of
practical verification tasks including
translation-validation~\cite{necula-tv},
superoptimization~\cite{massalin1987superoptimizer}, and program
synthesis~\cite{gulwani2017program}.
As we will see in \Cref{sec:case-studies}, the algorithm
that we develop for computing equivalence can also be straightforwardly
extended to other relational verification challenges, including one
inspired by the router bug above (c.f.\ our external filtering case
study in \Cref{subsec:utility-studies}).

There are several technical challenges related to mechanically and formally
proving protocol parser equivalence.
First, we need a computational model for parsers that is
expressive enough to handle practical parsers, but also sufficiently
restricted to enable tractable formal verification.
Second, we need efficient reasoning techniques based on symbolic
representations and domain-specific insights to handle the enormous
state space of real-world parsers.
Third, we need effective automation and tool support so programmers
can avoid manually crafting sprawling proofs of equivalence.

\paragraph{Certified Tooling.}
The foundational guarantees offered by proof assistants are highly
desirable in error-prone domains.
However, achieving these guarantees is notoriously hard, as proof
assistants need an experienced engineer's guidance to prove all but
the simplest goals.
One way to scale verification is to break systems into smaller
components that compose along shared specifications~\cite{deepspec}.
This allows individual verification tasks to be solved in isolation,
without compromising top-level guarantees.
Our hope is that push-button verifiers can reduce total proof burden by
certifying some components automatically.

Consider the task of proving that a realistic parser meets a
functional specification, in the style of VST~\cite{vstbook} proofs
which relate C programs to functional specifications.
The parser might be hand-optimized for performance reasons like the vectorized
parser in \Cref{fig:example-parser-syntax}, which makes it difficult to reason
about directly.
With an automated equivalence checker, we could justify replacing it
with the simpler and easier to verify reference implementation.
Other problems like translation validation~\cite{verified-lcm,crellvm}
or proof-producing synthesis~\cite{10.1145/3473589} would similarly
benefit from certified tooling.
%
%% However, automated verifiers often rely on complex optimizations and
%% clever uses of SMT solvers that can be hard to understand or debug.
% %
% Our approach relies on a SAT solver to dispatch simplified and rewritten
% verification conditions, but all its reductions and simplifications are
% certified.

% Commented out drawing of state machine
%
% \begin{figure}[tp]
%     \centering
%     \includegraphics[clip=true,trim=0pt 300pt 950pt 0,scale=0.09,page=8]{figures/implfigures.pdf}
%     \includegraphics[clip=true,trim=0pt 20pt 500pt 0,scale=0.09,page=9]{figures/implfigures.pdf}
%   \caption{\textbf{Reference} automaton (left) and \textbf{Vectorized} automaton
%   (right) for parsing MPLS followed by UDP. The double-outlined green circle is
%   the special accepting state. The normal parser states ($q_1$ through $q_5$)
%   are shown with each state's code in a box badged with the name of the state.
%   Unconditional transitions are indicated by bare arrows and conditional
%   transitions are indicated by arrows decorated with a guard expression.}
%   \label{fig:overview_example}
% \end{figure}

\paragraph{Our Contribution: \SYSTEM.}
We present \SYSTEM,\footnote{\url{https://github.com/verified-network-toolchain/leapfrog}} a new framework that addresses these challenges. It provides
an expressive automata model for parsers, with syntax inspired by
P4~\cite{p4paper,p4-spec}, a networking DSL\@. The model captures common
programming idioms and offers a domain-specific interface for packet parsing.
We demonstrate its applicability by encoding parsers for real-world
protocols like IPv4 and MPLS\@.

To establish the equivalence of \SYSTEM parsers, we extend classical
techniques based on bisimulations to work with symbolic
representations of the state space. We also develop a novel up-to
technique based on ``leaps'' that dramatically reduces the cardinality of the
constructed relation.
% We give formal proofs of the correctness of our approach.

We implement \SYSTEM as a Coq library. This allows us to mechanize our
metatheory and produce certificates.
Our algorithm, which runs inside the Coq prover, produces reusable Coq theorems of
parser equivalence.
% The library contains powerful techniques for constructing bisimulation proofs.
At a technical level, our Coq development combines classical
techniques based on predicate transformers, domain-specific
optimizations, and a plugin to interface with SMT solvers, to
facilitate effective automation. We apply \SYSTEM to several
benchmarks and find that it is able to scale up to handle realistic
protocols.

The contributions of this paper are as follows:
\begin{itemize}
\item We develop a parser model based on automata extended with
  domain-specific features (\Cref{sec:automata}).
\item We design efficient algorithms for establishing the equivalence
  of parsers based on symbolic and up-to techniques
  (\Cref{sec:equiv,sec:optimizing}).
\item We realize our approach in a Coq-based framework for
  automatically constructing equivalence proofs (\Cref{sec:impl}).
  Crucially, our design integrates off-the-shelf SMT solvers into a
  verification loop within Coq.
\item We explore \SYSTEM's expressiveness and scalability
  (\Cref{sec:case-studies}), finding that it can handle common protocol
  verification challenges and can perform translation validation for
  an existing parser compiler.
\end{itemize}
Overall, we believe \SYSTEM represents a promising step toward the
vision of certified proofs for protocol parsers.

\begin{figure}
\centering
\def\arraystretch{0}
\begin{tabular}{p{0.225\textwidth}:p{0.22\textwidth}}% chktex 44
\begin{lf-grey}
q1 {
  extract(mpls, 32);
  select(mpls[23:23]) {
    0 => q1
    1 => q2
  }
}
q2 {
  extract(udp, 64);
  goto accept
}
\end{lf-grey}
\vspace{3.5em}
\includegraphics[width=0.20\textwidth]{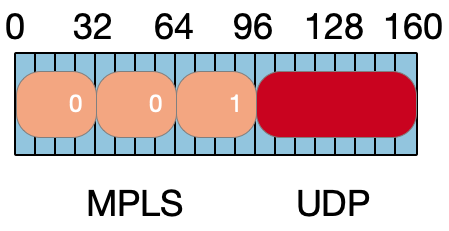}
&\begin{lf-grey}
q3 {
  extract(old, 32);
  extract(new, 32);
  select(old[23:23],
         new[23:23]) {
    (0, 0) => q3
    (0, 1) => q4
    (1, _) => q5
  }
}
q4 {
  extract(udp, 64);
  goto accept
}
q5 {
  extract(tmp, 32);
  udp := new ++ tmp;
  goto accept
}
\end{lf-grey}
\end{tabular}
\caption{Reference (q1, q2) and vectorized (q3, q4, q5) parsers for MPLS and UDP headers (depicted inset). }%
\label{fig:example-parser-syntax}
\vspace{-0.5em}
\end{figure}

%% file: overview.tex
\section{Overview}%
\label{sec:overview}

We now give a high-level overview of our automata model and
equivalence checking framework through an illustrative example.
Suppose we would like to parse packets with MPLS~\cite{rfc3031}
and UDP~\cite{rfc768} headers.
An MPLS header is a sequence of 32-bit \emph{labels}.
Rather than prefix this sequence with its length, the MPLS format marks the end
of the sequence with a label whose 24th bit is $1$.
This is analogous to the role of the null terminator in C strings.
In our example, the MPLS header is followed by an 8-byte UDP header.

We can parse this format in our \emph{P4 automata} model of parsers,
which is based on a subset of P4.
A P4 automaton is a state machine that parses a packet bitstring
into a collection of \emph{headers} stored in global variables, ultimately
either accepting or rejecting the packet.
Each state in a P4 automaton contains a program that may assign to variables and
extract some number of bits from the front of the packet into a variable.
For instance, an \texttt{extract(h, 64)} operation removes 64 bits from the front % chktex 36
of the packet and stores them into the header \texttt{h}.
Next, the machine transitions to a new state by branching on the
contents of its header variables.
This is accomplished with either an unconditional transition of the form
\texttt{goto state} or a conditional transition of the form \texttt{select(e) % chktex 36
\{ \(\overline{\texttt{pat => state}}\) \}}.
A \texttt{select} evaluates \texttt{e} and transitions to the state associated with the first matching pattern in \(\overline{\texttt{pat}}\).

We give two P4 automata for our format in \cref{fig:example-parser-syntax}.
The first automaton is a reference parser with a state \texttt{q1} that
parses MPLS and another state \texttt{q2} that parses UDP\@.
The second automaton has been \emph{vectorized} to parse two MPLS labels at
a time in state \texttt{q3}.
When the second label is the bottom of the stack, the vectorized parser
goes on to handle UDP normally (\texttt{q4}).
If the first label is the bottom of the stack, however, the vectorized
parser marshals the 32 bits of the ill-fated second label into UDP, along
with the remaining 32 bits of UDP header data remaining in the packet
(\texttt{q5}).

The two parsers in \cref{fig:example-parser-syntax} each use different names for
the header (e.g., \texttt{mpls} vs. \texttt{old}/\texttt{new}). Also, those
headers are overwritten multiple times---a real P4 parser would use an
array-like data structure called a \emph{header stack} to store the
labels~\cite[Section~8.17]{p4-spec}. Our language does not support header stacks
directly, although they can be emulated. Here, we omit this detail for
simplicity, and we focus on proving that the parsers accept the same sets of
packets. \SYSTEM can also be used to prove relational properties involving
the header values---see \Cref{sec:case-studies} for details.

\paragraph{Tractable Equivalence Checking.}
Our equivalence algorithm is inspired by Moore's
algorithm~\cite{hopcroft-1971} applied to the domain of P4 automata.
It is a worklist-style algorithm that begins with a coarse approximation of
language equivalence and iteratively refines it by analyzing the joint state
space of the two automata.
P4 automata have finitely many states and their header variables are
fixed-size bitvectors, so their configuration spaces are finite.
However, because of the large bitvectors encoded in their stores, P4 automata
may have an intractably large configuration space.
For instance, the automata in \cref{fig:example-parser-syntax} have a joint configuration space on the order of $2^{128} \approx 10^{38}$ states!
So, naive bisimulation-based approaches will never be tractable for
realistic automata.
We address this challenge by representing large relations \emph{symbolically},
rather than keeping track of concrete sets of configuration pairs.

Furthermore, we prune the configuration space from the start using a
simple reachability analysis.
This lets us avoid spurious search steps through unreachable configurations.

In the automata-theoretic semantics of the reference MPLS parser, the
$\mathtt{extract(udp, 64)}$ call performs $64$ steps that read one bit of
the packet into the buffer. The 64th step empties the buffer into the
\texttt{udp} header variable, and transitions to the accept state.
This bit-by-bit approach is needed to relate parsers that read the
packet in differently-sized chunks---as is the case for states \texttt{q1} and \texttt{q3} in \Cref{fig:example-parser-syntax}.
However, a naive search for a bisimulation that treats each step
separately would have a huge symbolic state and too many SMT
queries.
To counteract this, we introduce \emph{bisimulations with leaps},
which keeps the symbolic state compact and avoids redundant SMT
queries, processing multiple consecutive steps in each iteration.
Together, these optimizations make it feasible to compute
bisimulations for realistic parsers.

\begin{figure}
  \begin{small}
\(
\begin{array}{rcll}
  h &\in &H&\text{header names}\\
  n &\in &\mathbb{N}&\text{natural numbers}\\
  bv &\in &\bits^*&\text{bitvector}\\
  \syndef{\Sexp}{\Shdr}{headers}
  \syncase{\Sbits}{bitvectors}
  \syncase{\Sslice{\Sexp}{\Snat_1}{\Snat_2}}{bitslices}
  \syncase{\Sconcat{\Sexp_1}{\Sexp_2}}{concatenation}
  \syndef{\Spat}{\Sbits}{exact match}
  \syncase{\Sany}{wildcard}
  q &\in &Q \cup \{ \acc, \rej \}&\text{state names}\\
  \syndef{\Scase}{\Smany{\Spat} \Rightarrow q}{select case}
  \syndef{\Strans}{\Sgoto{q}}{direct}
  \syncase{\Ssel{\Smany{\Sexp}}{\Smany{\Scase}}}{select}
  \syndef{\Sop}{\Sextract{\Shdr}}{extract}
  \syncase{\Sassign{\Shdr}{\Sexp}}{assign}
  \syncase{\Sseq{\Sop_1}{\Sop_2}}{sequence}
  \syndef{\Sst}{\Sstate{q}{\Sop}{\Strans}}{states ($q \in Q$)} % chktex 1
  \syndef{\Saut}{\Smany{\Sst}}{P4 automaton}
\end{array}
\)
\end{small}
  \caption{Internal syntax for P4 automata.}\label{fig:internal-syntax}
\end{figure}

\paragraph{Certifying Parser Equivalence.}
% don't know where to put this paragraph so I'm just going to comment it out
%Vectorizing parsers to perform more operations in each state is a common
%optimization for programmable packet parsers.
%%
%Packet parsers have stringent bandwidth requirements but parsing hardware is
%highly resource-constrained, so compilers perform complex optimizations in order
%to emit fast code that still fits on the hardware.
%%
%Hardware limits may include bounds on the number of states, the number of
%different transitions, or the maximum number of bits that may be extracted from
%the packet in a single state.
%%
%When compilers fail to meet these limits, programmers resort to optimizing their
%code by hand.
%%
%In either case the resulting program, like our vectorized example, is harder to
%understand and harder to trust.
%
To make \SYSTEM usable in larger developments, we need it to produce
reusable proof certificates that can be checked by the Coq kernel.
Obtaining full certificates of equivalence from a push-button tool is
a significant engineering challenge.
Rather than write a solver in Coq for our verification conditions, which sit
roughly in the first-order theory of bitvectors, we chose to use external SMT
solvers.
This had engineering and performance benefits, but getting Coq and external
solvers to work together still posed several challenges.
First, existing interfaces between Coq and SMT solvers did not meet our needs.
We tried using existing plugins for proving Coq theorems with external
solvers~\cite{czajka-kaliszyk-2018,armand-etal-2011}, but found that
they scaled poorly or lacked support for key logical operators.
To address this, we developed a first-order theory of bitvectors
in Coq, as well as a plugin that pretty-prints this logic
in SMT-LIB format~\cite{barrett-etal-2010} and discharges the query using
an off-the-shelf solver.
We did not implement proof reconstruction, which converts the SMT
solver refutation into a Coq proof term.
Consequently the solver output and our pretty-printer must be
trusted, although this restriction could be lifted in future work.

Second, in order to target this low-level logic and improve the scalability of
our tool, we developed and verified a chain of compilation steps that go from the high-level
logic used in our algorithm to the low-level logic sent to solvers.
This process compiles away features like finite maps,
but also performs algebraic simplifications and other rewrites
to keep the size and complexity of our SMT queries under control.

Third, in order to communicate with SMT solvers, our algorithm needed to
perform I/O, which Coq programs are generally not allowed to do.
Coq tactics, however, are allowed to perform I/O, because the proof
they produce is still checked by the Coq kernel.
We rephrased our algorithm as a proof search problem (c.f.
\cref{fig:pre-bisim-coq}) and developed a custom Coq tactic as an
escape hatch, allowing the algorithm to consult an SMT solver while
still producing a Coq certificate (modulo the soundness of the solver
and our plugin).

%% file: automata.tex
\section{Parser Model}\label{sec:automata}

We now describe \emph{P4 automata} (P4As), an automata-theoretic model close to P4's parsing language~\cite{p4-spec}.
A P4A is a state machine that
  (1)~decides whether to accept a packet and
  (2)~builds a data structure (the \emph{store}) using the packet data.
The store consists of bitvectors called \emph{headers},\footnote{%
    Initial header values are undefined in P4~\cite[Sections~6.7
    and~8.22]{p4-spec}; our semantics considers them part of the packet.
}
and is a representation of the (partially) parsed packet used within the parser, but also in later processing phases.
If a P4A consumes data in each state, it terminates on finite packets.\footnote{
    Such a restriction is allowed by P4 specification~\cite[Sec.~12.10]{p4-spec}.
}

For example, state \texttt{q1} of the reference MPLS/UDP parser in \Cref{fig:example-parser-syntax} extracts 32 bits into the \texttt{mpls} header, before looping back to \texttt{q1} or transitioning to \texttt{q2} to parse a UDP header.

Concretely, a P4A is composed of states, each of which acts in two steps:
  first, it runs its internal program, which consumes some bits from the packet and updates the headers;
  then, it decides (based on the store) whether to accept (resp.\ reject) the
  packet by transitioning to the $\acc$ (resp.\ $\rej$) state, or continue processing the remainder elsewhere.
This second step defines the state's transitions.

\subsection{Syntax}

The syntax for P4As is best understood by example.
Consider the two programs in \cref{fig:example-parser-syntax}.
Together, these contain five states, named \texttt{q1} through \texttt{q5}.
Each state contains code, consisting of assignments and \texttt{extract} statements, and ending in a \texttt{select} or \texttt{goto} statement that defines the outgoing transitions.
Headers function as variables whose scope and lifetime is shared between states.
For instance, in \texttt{q5}, the vectorized parser extracts bits into \texttt{tmp}, and then stores the contents of \texttt{new} and \texttt{tmp} in \texttt{udp}, before accepting.

We formally define the syntax in \cref{fig:internal-syntax}, parameterized over a finite set of states $Q$, and a finite set of \emph{header names} $H$.
Each header $h \in H$ has an associated \emph{size} $\sfunc(h) \in \naturals^+$.
We refrain from specifying these parameters explicitly, as they can be inferred from the program text.
For instance, in the P4A on the left in \cref{fig:example-parser-syntax}, $Q = \{ \texttt{q1}, \texttt{q2} \}$, and $H = \{ \mathtt{mpls}, \mathtt{udp} \}$, while header sizes are $\sfunc(\mathtt{mpls}) = 32$ and $\sfunc(\mathtt{udp}) = 64$.

P4As associate with each state $q \in Q$ an operation block $\Sop(q)$ and transition block $\Strans(q)$.
Crucially, we require that least one call to \texttt{extract} appears in the operations of each state.
This guarantees that each state makes some progress on the packet, which ensures termination of both the parsing process and our equivalence checking algorithm.

\subsection{Semantics}

To provide a semantics for P4As, we first assign a semantics to the operation and transition code associated with each state.
We fix a P4A $\Saut$ with states $Q$, headers $H$ and sizes $\sfunc$.
We write $|\Sbits|$ for the length of $\Sbits \in \bits^*$.
We define $S$ as the finite set of functions $s: H \to \bits^*$ where $|s(h)| = \sfunc(h)$.

We first give a semantics to expressions.

\begin{definition}[Expression Semantics]%
\label{def:expression-semantics}
Let $w, x \in \bits^*$.
We write $wx$ for their concatenation.
If $\Snat_1, \Snat_2 \in \naturals$, we write $\Sslice{w}{\Snat_1}{\Snat_2}$ for the zero-indexed substring starting at position $\min(\Snat_1, |w|-1)$ and ending at $\min(\Snat_2, |w|-1)$, inclusive.

Given an expression $\Sexp$, we inductively define its semantics in the form of a function $\semexp{\Sexp}: S \to \bits^*$, as follows:
\begin{align*}
\semexp{\Shdr}(s) &= s(h) &
  \semexp{\Sslice{\Sexp}{\Snat_1}{\Snat_2}}(s) &= \Sslice{\semexp{\Sexp}(s)}{\Snat_1}{\Snat_2} \\
\semexp{\Sbits}(s) &= \Sbits &
  \semexp{\Sconcat{\Sexp_1}{\Sexp_2}}(s) &= \semexp{\Sexp_1}(s) \semexp{\Sexp_2}(s)
\end{align*}
There is a straightforward typing judgement $\typexpbare$, where $\typexp{\Sexp}{\Snat}$ implies that $|\semexp{\Sexp}(s)| = \Snat$.
We elide this definition.
\end{definition}

Next, we give a semantics to operations and transitions, which constitute the code that can appear inside states.

\begin{definition}[Operation Semantics]
When $s \in S$, $h \in H$ and $v \in \bits^{\sfunc(h)}$, we use $s[v/h]$ to denote the store where $s[v/h](h) = v$, and $s[v/h](h') = s(h')$ for all $h' \in H\setminus\{h\}$.

For operations $\Sop$, define $\opsize{\Sop} \in \naturals$ inductively, as follows:
\begin{mathpar}
\opsize{\Sassign{\Shdr}{\Sexp}} = 0
\and
\opsize{\!\Sextract{h}} = \sfunc(h) \vspace{-2mm}
\\
\opsize{\Sseq{\Sop_1}{\Sop_2}} = \opsize{\Sop_1} + \opsize{\Sop_2}
\end{mathpar}
Intuitively, $\opsize{\Sop}$ is the exact number of bits necessary to execute all \texttt{extract} statements that appear in $\Sop$.

For each block of operations $\Sop$, we define a \emph{partial} function $\semop{\Sop}: S \times \bits^* \rightharpoonup S \times \bits^*$, as follows:
\begin{align*}
\semop{\Sassign{\Shdr}{\Sexp}}(s, w) &= \angl{s[v/h], w} \tag{if \ensuremath{\semexp{\Sexp}(s) = v} and \ensuremath{|v| = \sfunc(h)}}\\
\semop{\Sextract{\Shdr}}(s, xy) &= \angl{s[x/h], y} \tag{if \ensuremath{|x| = \sfunc(h)}} \\
\semop{\Sseq{\Sop_1}{\Sop_2}}(s, w) &= \semop{\Sop_2}(\semop{\Sop_1}(s, w))
\end{align*}
There exists a type judgement $\typopbare$ such that if $\typop{\Sop}$ then $\semop{\Sop}(s, w) \in S \times \{ \epsilon \}$ for all $\Sbits \in \bits^{\opsize{\Sop}}$.
\end{definition}

\begin{definition}[Pattern and Transition Semantics]
For a pattern $\Spat$, define $\sempat{\Spat} \subseteq \bits^*$ by case distinction:
\begin{mathpar}
\isem{\Sbits}_\mathcal{P} = \{ \Sbits \}
\and
\isem{\Sany}_\mathcal{P} = \bits^*
\end{mathpar}
Given a transition block $\Strans$, we define a partial function $\semtrans{\Strans}: S \rightharpoonup Q \cup \{ \acc, \rej \}$ inductively, as follows:
\begin{mathpar}
\isem{\Sgoto{q}}_\mathcal{T}(s) = q
\and
\semtrans{\Ssel{\Smany{\Sexp}}{}}(s) = \rej
\and
\inferrule{%
    \forall i.\ \isem{\Sexp_i}_\mathcal{E}(s) = v_i \\
    q' = \semtrans{\Ssel{\Smany{\Sexp}}{\Smany{\Scase}}}(s)
}{%
    \semtrans{\Ssel{\Smany{\Sexp}}{\Smany{\Spat} \Rightarrow q;\ \Smany{\Scase}}}(s) =
        {\begin{cases}
        q & \forall i.\ v_i \in \sempat{\Spat_i} \\
        q' & \text{otherwise}
        \end{cases}}
}
\end{mathpar}
As before, a straightforward type judgement $\typtransbare$ can be formulated such that $\typtrans{\Strans}$ implies that $\semtrans{\Strans}(s)$ is defined.
\end{definition}

We now have the ingredients necessary to define the dynamics of a P4A in terms of a deterministic finite automaton (DFA).
To facilitate the comparison of P4As that consume packets in differently-sized chunks, this DFA buffers until it has read enough bits to execute the \texttt{extract} blocks associated with the current state.
We first precisely define a configuration of a P4A, as follows.

\begin{definition}[Configurations]
A \emph{configuration} is a triple
\[
  \angl{q, s, w} \in (Q \cup \{ \acc, \rej \}) \times S \times \bits^*.
\]
where $|w| < \opsize{\Sop(q)}$ if $q \in Q$, and $w = \epsilon$ otherwise.
We write $C$ for the (finite) set of configurations, and $F$ for the \emph{accepting configurations}: $\{ \angl{\acc, s, \epsilon} \in C : s \in S \}$.
\end{definition}

We can define a bit-by-bit step function on configurations, which implements the idea of filling up the store before actuating the transition, outlined above.
% This function implements the idea outlined above: fill up the third component of the configuration until enough bits are gathered, at which point the associated operations can be evaluated, and the next state can be determined.
\begin{definition}[Configuration Dynamics]
We define the step function $\delta: C \times \bits \rightharpoonup C$ as follows.
Let $c = \angl{q, s, w} \in C$.
If $q \in Q$, then we define $\delta(c, b)$ by setting
\begin{mathpar}
\delta(c, b) =
    \begin{cases}
    \angl{q, s, wb} & |wb| < \opsize{\Sop(q)} \\[1mm]
    \angl{\semtrans{\Strans(q)}(s'), s', \epsilon} & \semop{\Sop(q)}(s, wb) = \angl{s', \epsilon} \\[1mm]
    \end{cases}
\end{mathpar}
Otherwise, if $q \in \{ \acc, \rej \}$, then $\delta(c, b) = \angl{\rej, s, \epsilon}$.

There exists a type judgement $\typautbare$ such that $\typaut{\Saut}$ implies that $\delta$ is well-defined and total; again, we omit its definition.%
\footnote{
    Our requirement that each state extracts some bits is part of this typing judgement, and in fact necessary in order for the definition of $\delta$ to be useful.
    Because a transition is triggered by the final bit, if $\opsize{\Sop(q)} = 0$ for some state, then there would be no way to actuate this transition.
}
\end{definition}

To match the behavior of P4 parsers, accepting states should not parse any
further input. As a consequence a configuration of the form $\angl{\acc, s,
\epsilon}$ steps unconditionally to $\angl{\rej, s, \epsilon}$.

Put together, $\angl{C, \delta, F}$ is a DFA\@.
We can therefore define the language semantics of our parser $\Saut$ as a function $\semaut{\Saut}: Q \times S \to 2^{\bits^*}$, where $2^X$ denotes the set of subsets of a set $X$.
This semantics associates with each initial state and store the set of bit-strings that lead to an accepting configuration.

\begin{definition}[Multi-Step Configuration Dynamics]
We can lift $\delta$ to $\delta^*: C \times \bits^* \rightharpoonup C$ as follows:
\begin{mathpar}
\delta^*(c, \epsilon) = c
\and
\delta^*(c, bw) = \delta^*(\delta(c, b), w)
\end{mathpar}
Given $c \in C$, we define its \emph{language} $L(c) \subseteq \bits^*$ as follows:
\[
    L(c) = \{ w \in \bits^* : \delta^*(c, w) \in F \}
\]
Given $q \in Q$ and $s \in S$, we define $\semaut{\Saut}(q, s) = L(q, s, \epsilon)$.
\end{definition}

Our semantics embeds the initial store in the start state.
Our equivalence checking procedure can help verify that packet acceptance does not depend on the initial store value.

%% file: equivalence.tex
\section{Symbolic Equivalence Checking}\label{sec:equiv}

Many verification questions about P4As can be phrased as questions about the underlying DFAs.
For instance, let $\Saut_1$ and $\Saut_2$ be the P4As from \cref{fig:example-parser-syntax}, suppose we want to verify that they accept the same packets when started from certain initial states \texttt{q1} and \texttt{q3}, regardless of their initial store.
To do this, we could check whether $L(\mathtt{q1}, s_1, \epsilon) = L(\mathtt{q3}, s_2, \epsilon)$ for all $s_1, s_2 \in S$.
This problem is decidable, because $S$ is finite and language equivalence of DFAs is decidable~\cite{moore-1956}.

Unfortunately, the DFA arising from a P4A may be extremely large: every $q \in Q$ contributes $|S| \times 2^{\opsize{\Sop(q)}-1}$ configurations.
Even for simple parsers, this leads to an intractably large configuration space.
For instance, for the reference MPLS parser on the left in \cref{fig:example-parser-syntax}, $|S| = 2^{96}$; a back-of-the-envelope calculation then tells us that $|C| \geq 10^{38}$.

Moreover, we anticipate that a large portion of the configuration space is reachable, and should therefore be taken into consideration.
This is because parsers tend to propagate every bit of the packet into the store in order to facilitate packet reconstruction for forwarding.
Off-the-shelf algorithms for DFAs are therefore unlikely to scale to this setting.

In this section, we develop an algorithm that can answer several questions about P4As.
This algorithm mitigates state space explosion by representing configurations symbolically.
Our presentation focuses on deciding language equivalence of configurations.
As a consequence, the procedure can be thought of as a variation on Moore's algorithm~\cite{moore-1956}.
We discuss more general applications in Section~\ref{sec:case-studies}.

For the sake of simplicity, we fix a P4A $\Saut$ with underlying DFA $\angl{C, \delta, F}$ for the remainder of this section.
One can compare configurations in two different P4As by taking their disjoint sum, renaming states and headers as necessary.

\subsection{A Symbolic Approach}

A sound and complete method to show that two configurations of our DFA $\angl{C, \delta, F}$ accept the same language is to demonstrate that they are related by a \emph{bisimulation}~\cite{hopcroft-karp-1971}, i.e., a relation $R \subseteq C \times C$ such that when $c_1 \mathrel{R} c_2$,
    (1)~$c_1 \in F$ if and only if $c_2 \in F$; and
    (2)~$\delta(c_1, b) \mathrel{R} \delta(c_2, b)$ for all $b \in \bits$.

A language equivalence checking algorithm for DFAs typically tries to build some form of bisimulation.
Because $C$ may be very large, representing a bisimulation by listing its constituent pairs becomes intractable quickly.
Luckily, we can write down bisimulations symbolically.

\begin{example}%
\label{ex:lots-of-pairs}
Suppose $\Saut$ is the disjoint sum of the MPLS parsers displayed in \cref{fig:example-parser-syntax}.
Let $R$ be the smallest relation on $C$ satisfying the following rules for all $s_1, s_2 \in S$:
\[
    \inferrule{%
        w \in \bits^{32} \quad
        x \in \bits^* \quad
        |x| < 32
    }{%
        \angl{\mathtt{q2}, s_1, wx} \mathrel{R} \angl{\mathtt{q5}, s_2, x}
    }
    \quad
    \inferrule{%
        q \in \{ \acc, \rej \}
    }{%
        \angl{q, s_1, \epsilon} \mathrel{R} \angl{q, s_2, \epsilon}
    }
\]
$R$ is a bisimulation, and thus all configurations related by $R$ have the same language.
Clearly, this representation is much more concise than listing the contents of $R$ explicitly.
\end{example}

\begin{figure}
  \begin{small}
  \[
    \begin{array}{rcll}
      x &\in& \Var & \text{variables}\\
      \syndef{\Sbexp}{\Sbits}{literal}
      \syncase{\Sbbuf{<}, \Sbbuf{>}}{left and right buffer}
      \syncase{\Sbhdr{<}, \Sbhdr{>}}{left and right header}
      \syncase{\Svar}{variable}
      \syncase{\Sslice{\Sbexp}{\Snat_1}{\Snat_2}}{slice}
      \syncase{\Sconcat{\Sbexp_1}{\Sbexp_2}}{concat}
      \syndef{\Spred}{\Sbexp_1 = \Sbexp_2}{bitvector equality}
      \syncase{\Sstatepred{q}{<}, \Sstatepred{q}{>}}{left and right state assertion}
      \syncase{\Sbuflen{\Snat}{<}, \Sbuflen{\Snat}{>}}{left and right buffer length}
      \syndef{\phi}{\bot}{bottom}
      \syncase{\Spred}{atomic predicate}
      \syncase{\phi_1 \implies \phi_2}{implication}
    \end{array}
  \]
  \end{small}
  \caption{Syntax for relations on configurations.}%
  \label{fig:confrel-syntax}
\end{figure}

To systematically represent and manipulate symbolic relations on configurations, we propose the syntax in \cref{fig:confrel-syntax}.
Its formulas are generated by equality assertions between expressions built over the buffers and stores of both configurations, as well as predicates about states and buffer lengths.
We also include variables $x \in \Var$ for later use. We omit conjunction
($\vee$) and disjunction ($\wedge$) from the syntax to keep our
definitions brief. They are derivable from $\implies$
and $\bot$, so we will still use them in the sequel as abbreviations.

\begin{example}%
\label{ex:symbolic-bisimulation}
% (i) The formula $q_1^< \wedge \#0^< \wedge q_2^> \wedge \#0^>$ relates every configuration of the form $\angl{q_1, s_1, \epsilon}$ to every configuration of the form $\angl{q_2, s_2, \epsilon}$
We can describe the pairs matching the second rule from \cref{ex:lots-of-pairs} using the following formula
\[
    \phi = (\Sstatepred{\acc}{<} \wedge \Sstatepred{\acc}{>}) \vee (\Sstatepred{\rej}{<} \wedge \Sstatepred{\rej}{>})
\]
Given $n < 32$, we can choose the formula $\phi_n$ to symbolize the first rule from \cref{ex:lots-of-pairs} where $|x| = n$:
\[
    \phi_n = \Sbuflen{(n+32)}{<} \wedge \Sstatepred{q2}{<} \wedge \Sbuflen{n}{>} \wedge \Sstatepred{q5}{>} \wedge \Sslice{\Sbbuf{<}}{0}{31} = \Sbbuf{>}
\]
In total, $R$ is represented by the formula $\phi \vee \phi_0 \vee \cdots \vee \phi_{31}$.
\end{example}

\begin{definition}
A \emph{valuation} is a function $\sigma: \Var \to \bits$.
For every bitvector expression $\Sbexp$ and valuation $\sigma$, we define $\isem{\Sbexp}_{\mathcal{B}}^\sigma: C \times C \to \bits^*$ inductively as follows, where $c^<, c^> \in C$ are such that $c^\lessgtr = \angl{q^\lessgtr, s^\lessgtr, w^\lessgtr}$ for ${\lessgtr} \in \{ <, > \}$:
\begin{mathpar}
\bsem{\Sbits}{\sigma}{c^<, c^>} = \Sbits \and
\bsem{\Sbbuf{\lessgtr}}{\sigma}{c^<, c^>} = w^\lessgtr \and
\bsem{\Svar}{\sigma}{c^<, c^>} = \sigma(\Svar) \and
\bsem{\Sbhdr{\lessgtr}}{\sigma}{c^<, c^>} = s^\lessgtr(h) \and
% \bsem{\Sslice{\Sbexp}{\Snat_1}{\Snat_2}}{\sigma}{c^<, c^>} = \Sslice{\bsem{\Sbexp}{\sigma}{c^<, c^>}}{\Snat_1}{\Snat_2} \and
% \bsem{\Sconcat{\Sbexp_1}{\Sbexp_2}}{\sigma}{c^<, c^>} =
%     \bsem{\Sbexp_1}{\sigma}{c^<, c^>}
%     \bsem{\Sbexp_2}{\sigma}{c^<, c^>}
\end{mathpar}
The cases for slices and concatenation are as in \cref{def:expression-semantics}.

For a formula $\phi$ and valuation $\sigma$, define $\lsem{\phi}{\sigma}$ as the least relation on $C$ satisfying the following rules for all $c^<, c^> \in C$, where $c^\lessgtr = \angl{q^\lessgtr, s^\lessgtr, w^\lessgtr}$, and $n^\lessgtr = |w^\lessgtr|$ for ${\lessgtr} \in \{ <, > \}$:
\begin{mathpar}
\begin{array}{c}
c^< \lsem{\Sstatepred{q}{<}}{\sigma} c^> \\[2mm]
c^< \lsem{\Sstatepred{q}{>}}{\sigma} c^> \\[2mm]
c^< \lsem{\Sbuflen{n}{<}}{\sigma} c^> \\[2mm]
c^< \lsem{\Sbuflen{n}{>}}{\sigma} c^>
\end{array} \and
\begin{array}{c}
\inferrule{%
    \isem{\Sbexp_2}_{\mathcal{B}}(c_1, c_2, \sigma) = \isem{\Sbexp_2}_{\mathcal{B}}(c_1, c_2, \sigma)
}{%
    c_1 \lsem{\Sbexp_1 = \Sbexp_2}{\sigma} c_2
} \\[6mm]
\inferrule{%
    c_1 \lsem{\phi_1}{\sigma} c_2 \implies c_1 \lsem{\phi_2}{\sigma} c_2
}{%
    c_1 \lsem{\phi_1 \implies \phi_2}{\sigma} c_2
}
\end{array}
\end{mathpar}
Let $\phi$ and $\psi$ be formulas.
We write $\sem{\phi}_\mathcal{L}$ for the relation on $C$ where $c_1 \sem{\phi}_{\mathcal{L}} c_2$ if and only if $c_1 \psem{\phi}{\mathcal{L}}{\sigma} c_2$ for all valuations $\sigma$.
Finally, we write $\phi \vDash \psi$ when ${\sem{\phi}_{\mathcal{L}}} \subseteq {\sem{\psi}_{\mathcal{L}}}$.
\end{definition}

Note that because there are finitely many configurations and valuations, entailments are decidable.
We will revisit this particular decision problem in Sections~\ref{sec:optimizing} and~\ref{sec:impl}.

We can now define symbolic bisimulations, as follows.

\begin{definition}[Symbolic Bisimulation]
A \emph{symbolic bisimulation} is a formula $\phi$ such that $\sem{\phi}_\mathcal{L}$ is a bisimulation.
\end{definition}

Finding a (symbolic) bisimulation is a sound and complete method to establish language equivalence of states.

\begin{restatable}{lemma}{restatesymbolicbisimulationsoundcomplete}%
\label{lemma:symbolic-bisimulation-sound-complete}
For formulas $\phi$, the following are equivalent:
\begin{enumerate}
    \item
    There exists a symbolic bisimulation $\psi$ such that $\phi \vDash \psi$.

    \item
    There exists a bisimulation $R$ such that ${\sem{\phi}_\mathcal{L}} \subseteq R$.

    \item
    If $c_1 \sem{\phi}_\mathcal{L} c_2$, then $L(c_1) = L(c_2)$.
\end{enumerate}
\end{restatable}

\begin{algorithm}[t]
\DontPrintSemicolon%
\caption{Symbolic equivalence checking.}\label{algorithm:symbolic-refinement}%
\label{alg:abstract-symbolic-refinement}

\Input{A formula $\phi$ representing initial states.}
\Input{A set of formulas $I$ s.t.\ for all $c_1, c_2 \in C$,
    \[
        \left\lbrack \forall \psi \in I.\ c_1 \sem{\psi} c_2 \right\rbrack
        \Leftrightarrow \lbrack c_1 \in F \Leftrightarrow c_2 \in F \rbrack
    \]
}
\Input{%
    A function $\WP$ s.t.\ for all $\psi$, and $c_1, c_2 \in C$,
    \[
        \hspace{-3ex}\left\lbrack \forall b \in \bits.\ \delta(c_1, b) \sem{\psi} \delta(c_2, b) \right\rbrack
        \Leftrightarrow c_1 \sem{\bigwedge \WP(\psi)} c_2
    \]
}
\Output{{\True} if and only if for all $c_1, c_2 \in C$ with $c_1 \sem{\phi}_\mathcal{L} c_2$, it holds that $L(c_1) = L(c_2)$}
\BlankLine%

$R \gets \emptyset;\; T \gets I$\;
\While{$T \neq \emptyset$}{
    pop $\psi$ from $T$\;
    \If{{\normalfont\textbf{not}} $\bigwedge R \vDash \psi$}{\label{line:entailment}%
        $R \gets R \cup \{ \psi \}$\;
        $T \gets T \cup \WP(\psi)$\;
    }
}
\KwRet{\rm $\True$ if $\phi \vDash \bigwedge R$, otherwise $\False$}\;
\end{algorithm}

\subsection{The Weakest Symbolic Bisimulation}

To search for a symbolic bisimulation, we turn to Moore's algorithm~\cite{moore-1956}.
In its concrete formulation, this algorithm approximates the largest (coarsest) bisimulation from above, by iteratively removing non-bisimilar pairs.
Eventually, the process stops, at which point the remaining pairs must be bisimilar; hence, the computed relation is the largest bisimulation.
Two configurations are related by some bisimulation if and only if they are related by the largest bisimulation, so the algorithm concludes with a simple containment check.

Moore's algorithm can be made symbolic, by representing the current overapproximation as a formula and successively strengthening it, thus converging to the weakest symbolic bisimulation.
We present an abstract formulation of this process in \cref{algorithm:symbolic-refinement}.
The algorithm has two parameters.
\begin{itemize}
    \item
    The formulas in $I$ constitute the initial overapproximation of the weakest symbolic bisimulation.
    In Section~\ref{sec:case-studies}, we consider instantiations of $I$ that can be used to verify different but related properties.
    \item
    The function $\WP$ takes a formula $\phi$, and outputs a set of formulas
    whose conjunction represents a \emph{weakest precondition} of $\phi$, in the
    sense that two configurations are related by \emph{all} formulas in
    $\WP(\phi)$ if and only if they step into configurations related by $\phi$.
\end{itemize}

\noindent
\cref{algorithm:symbolic-refinement} builds the weakest symbolic bisimulation as a set of conjuncts $R$, maintaining a frontier $T$ of formulas to be considered.
The frontier is initially $I$.
In each iteration, we pop a conjunct $\psi$ from $T$ and check if it is entailed by $\bigwedge R$.
If $\bigwedge R \not\vDash \phi$, then $\psi$ constitutes a novel restriction, and we add it to $R$.
Because bisimulations are closed under steps, we add the weakest preconditions of $\psi$ to $T$, to be checked later.
If $\bigwedge R \vDash \psi$, including $\psi$ in $R$ would not change $\bigwedge R$, so we move on.
The loop terminates when $T$ is empty; at this point, $\bigwedge R$ will be the weakest symbolic bisimulation, and the algorithm checks $\phi \vDash \bigwedge R$.
We instantiate the parameters momentarily; first, we show that the algorithm is correct.

\begin{restatable}{theorem}{restatesymbolicrefinementcorrect}%
\label{theorem:symbolic-refinement-correct}
\cref{algorithm:symbolic-refinement} is correct.
\end{restatable}
\begin{proof}[Proof Sketch]
For termination, note that in each iteration either $\sem{\bigwedge R}$ or $T$ shrinks; hence, the algorithm must terminate.

For partial correctness, one can show the following invariants:
  (1)~if $\phi$ is a symbolic bisimulation, then $\phi \vDash \bigwedge R \wedge \bigwedge T$; and
  (2)~configurations related by $\bigwedge R \wedge \bigwedge T$ are equally accepting, and
  (3)~configurations related by $\bigwedge R \wedge \bigwedge T$ step into configurations related by $\bigwedge R$.
Thus, when \cref{algorithm:symbolic-refinement} terminates, $\bigwedge R$ must be the largest symbolic bisimulation, and we can conclude by applying \cref{lemma:symbolic-bisimulation-sound-complete}.
\end{proof}

\subsection{Instantiating the Parameters}
We now sketch how we instantiate the parameters of \cref{algorithm:symbolic-refinement}; the details are worked out in our Coq development.

For $\WP$, the main idea is to focus on a particular subclass of formulas.
First, we isolate assertions about the current state; this lets us calculate weakest preconditions on a state-by-state basis, by means of a traditional substitution-based procedure on the formula using the associated program text.
Second, we isolate statements about buffer lengths; this means that when a formula in our algorithm makes a claim about the buffer contents, it does so in a context where the buffer length is known.
This simplifies the analysis and generation of formulas, because we do not have to cover cases where slices go beyond the end of a bitvector.

Concretely, this format takes the following form.

\begin{definition}[Templates]
A \emph{template} is a pair $\angl{q, n} \in (Q \cup \{ \acc, \rej \}) \times \naturals$ where $n < \opsize{\Sop(q)}$ if $q \in Q$, and $n = 0$ otherwise.
The set of all templates is $T$.
When $t = \angl{q, n}$ and ${\lessgtr} \in \{ <, > \}$, we write $t^\lessgtr$ as shorthand for $q^\lessgtr \wedge n^\lessgtr$.

A formula $\phi$ is \emph{pure} when it does not contain state or buffer length assertions; $\phi$ is \emph{template-guarded} if it is of the form $t_1^< \wedge t_2^> \implies \psi$ where $t_1, t_2 \in T$ and $\psi$ is pure.
\end{definition}

Let $\phi$ be template-guarded.
We compute $\WP(\phi)$ by operating on the left- and right hand side, giving rise to functions $\WP^<$ and $\WP^>$, whose definitions we omit.
Each takes a formula and a pair of state templates, as well as a fresh variable $\Svar \in \Var$ to represent the bit to be read, and returns a formula.
The relevant correctness statement is as follows.

\begin{lemma}%
\label{lemma:wp-asymmetric-correctness}
Let $\phi$ be pure, let $\Svar \in \Var$ be fresh for $\phi$, and let $c^<, c^> \in C$ as well as $t \in T$.
The following are equivalent:
\begin{enumerate}
    \item For all $b \in \bits$, we have $\delta(c^<, b) \sem{t^< \implies \phi}_\mathcal{L} c^>$.
    \item For all $t' \in T$, $c^< \sem{t'^< \implies \WP^<(\phi, t', t, x)}_\mathcal{L} c^>$.
\end{enumerate}
A similar equivalence holds for $\WP^>$.
Furthermore, if $\phi$ is pure, then so are $\WP^<(\phi, x, t', t)$ and $\WP^>(\phi, x, t', t)$.
\end{lemma}

Using $\WP^<$ and $\WP^>$, we can then provide a version of $\WP$ that acts on and returns template-guarded formulas.
Its definition and correctness statement is as follows.

\begin{lemma}%
\label{lemma:wp-correctness}
Let $t_1^< \wedge t_2^> \implies \phi$ be template-guarded, and let $\Svar$ be fresh in $\phi$.
Define $\WP(t_1^< \wedge t_2^> \implies \phi)$ as the smallest set satisfying the following rule, for all $t_1', t_2' \in T$:
\[
    \inferrule{%
        \phi' = \WP^<(\WP^>(\phi', x, t_2', t_2), x, t_1', t_1) \\
    }{%
        [t_1'^< \wedge t_2'^> \implies \phi'] \in \WP(t_1^< \wedge t_2^> \implies \phi)
    }
\]
Now $\WP$ fits the requirement from \cref{algorithm:symbolic-refinement}, when restricted to template-guarded formulas.
Moreover, each of the formulas in $\WP(t_1^< \wedge t_2^> \implies \phi)$ is template-guarded.
\end{lemma}
% \begin{proof}[Proof sketch]
% Can be shown by applying the correctness result about $\WP^<$ and $\WP^>$ from \cref{lemma:wp-asymmetric-correctness} twice.
% \end{proof}

By the latter property, if all formulas in $I$ are template-guarded, then the formulas in $R$ and $T$ remain template-guarded.
We thus instantiate $I$ as a set of template-guarded formulas that rule out pairs containing both accepting and non-accepting configurations, as follows.

\begin{lemma}
Let $t_\acc = \angl{\acc, 0}$.
Define $I$ as the smallest set of formulas satisfying the following rule:
\[
    \inferrule{%
        t_1, t_2 \in T \\
        t_1 = t_\acc \iff t_2 \neq t_\acc
    }{%
        [t_1^< \wedge t_2^> \implies \bot] \in I
    }
\]
Now $I$ fits the requirement from \cref{algorithm:symbolic-refinement}.
\end{lemma}

%% file: optimizing.tex
\section{Optimizing the Algorithm}%
\label{sec:optimizing}

We now discuss two optimizations of \cref{algorithm:symbolic-refinement}.
The first optimization refines $\WP$ and $I$ such that fewer entailments between formulas need to be checked (line~\ref{line:entailment}).
The second optimization generalizes $\WP$ to compute multi-step weakest preconditions, thereby strengthening the approximation of the weakest symbolic bisimulation more quickly.

\subsection{Abstract Interpretation}

\cref{algorithm:symbolic-refinement} computes the weakest symbolic bisimulation, which relates \emph{all} language equivalent configurations, but it cares only about the configurations related by $\phi$.
We can compute a symbolic bisimulation more loosely, by disregarding unreachable (and hence, irrelevant) configuration pairs.

\begin{example}
Recall the symbolic bisimulation in \cref{ex:symbolic-bisimulation}, which was sufficient to conclude language equivalence of related configurations.
There was no need to compute the largest symbolic bisimulation, which involves many configuration pairs unreachable from the pairs of interest.
\end{example}

Of course, computing the set of reachable pairs---even symbolically---is tantamount to checking equivalence.
Instead, we approximate it by analyzing the P4A to capture the pairs of reachable configurations based on their templates.

To this end, let $\rho(\Strans)$ denote the set of states appearing in a transaction block $\Strans$.
We define $\sigma: T \to 2^T$ as follows:
\[
    \sigma(q, n) =
        \begin{cases}
        \{ \angl{q, n+1} \} & q \in Q \wedge n + 1 < \sfunc(q) \\
        \rho(\Strans(q)) \times \{0 \} & q \in Q \wedge n + 1 = \sfunc(q) \\
        \{ \angl{\rej, 0} \} & q \in \{ \acc, \rej \}
        \end{cases}
\]

When $c = \angl{q, s, w} \in C$, write $\floor{c}$ for $\angl{q, |w|} \in T$, i.e., the unique template describing $c$.
One can show that for all $c \in C$ and $b \in \bits$, we have $\floor{\delta(c, b)} \in \sigma(\floor{c})$.
In a sense, this makes $\sigma$ an abstract interpretation of $\delta$.

Given a formula $\phi$, we define $\reach_\phi$ as the smallest relation on $T$ satisfying the following rules:
\begin{mathpar}
\inferrule{%
    c_1 \sem{\phi}_\mathcal{L} c_2
}{%
    \floor{c_1} \mathrel{\reach_\phi} \floor{c_2}
}
\and
\inferrule{%
    t_1 \mathrel{\reach_\phi} t_2 \\
}{%
    \sigma(t_1) \times \sigma(t_2) \subseteq \reach_\phi
}
\and
\end{mathpar}
Usually, the pairs generated by the first rule can be inferred from $\phi$.
For instance, if we want to compare the languages of two initial states $q_1$ and $q_2$, then $\phi = q_1^< \wedge 0^< \wedge q_2^> \wedge 0^>$, and so the sole instantiation of the first rule yields $\angl{q_1, 0} \mathrel{\reach_\phi} \angl{q_2, 0}$.
Computing the full contents of $\reach_\phi$ is then a matter of applying the second rule until a fixpoint is reached.

\begin{restatable}{theorem}{restatesymbolicrefinementinitializecorrect}%
\label{theorem:symbolic-refinement-initialize-correct}
Let $\phi$ be a formula.
\cref{algorithm:symbolic-refinement} remains correct for this $\phi$ if we set $I$ to the smallest set satisfying the rule
\[
    \inferrule{%
        t_1 \mathrel{\reach_\phi} t_2 \\
        t_1 = t_\acc \iff t_2 \neq t_\acc
    }{%
        [t_1^< \wedge t_2^> \implies \bot] \in I
    }
\]
and for each template-guarded formula $t_1^< \wedge t_2^> \implies \psi$ we set $\WP(t_1^< \wedge t_2^> \implies \psi)$ to the smallest set satisfying the rule
\[
    \inferrule{%
        t_1' \mathrel{\reach_\phi} t_2' \\
        \psi' = \WP^<(\WP^>(\psi', x, t_2', t_2), x, t_1', t_1) \\
    }{%
        [t_1'^< \wedge t_2'^> \implies \phi'] \in \WP(t_1^< \wedge t_2^> \implies \phi)
    }
\]
where $\Svar \in \Var$ is some variable that is fresh for $\psi$.
\end{restatable}
% \begin{proof}[Proof sketch]
% We adapt the proof from \cref{theorem:symbolic-refinement-correct}.
% The termination argument is the same.
% For partial correctness, let $\psi = \bigvee_{t_1 \mathrel{\reach_\phi} t_2} t_1^< \wedge t_2^>$, and note that $\phi \vDash \psi$.
% With suitable adaptations to the invariants, one can show that when the main loop terminates, $\psi \wedge \bigwedge R$ is the weakest symbolic bisimulation that entails $\psi$.
% Thus, when the algorithm returns \textbf{true}, the pairs related by $\phi$ indeed have the same language.
% Conversely, if the pairs related by $\phi$ have the same language, then $\phi \vDash \rho$ for some symbolic bisimulation $\rho$; since $\rho \wedge \psi$ is a bisimulation that entails $\psi$, the algorithm returns \textbf{true}.
% \end{proof}

\subsection{Leaps and Bounds}

\cref{algorithm:symbolic-refinement} operates on a bit-by-bit basis.
However, most steps just fill up the buffer, and do not affect the state or store.
We exploit this to compute a different form of weakest precondition, which takes as many steps as necessary to execute a ``real'' state-to-state transition in the P4A\@.

The following auxiliary notion allows us to compute the number of steps until the next transition.

\begin{definition}[Leap Size]
Let $c_1, c_2 \in C$ and $c_i = \angl{q_i, s_i, w_i}$; we define the \emph{leap size} $\sharp(c_1, c_2) \in \naturals$ as follows:
\[
    \sharp(c_1, c_2) =
        \begin{cases}
        1 & q_1, q_2 \not\in Q \\
        \opsize{\Strans(q_1)} - |w_1| & q_1 \in Q, q_2 \not\in Q \\
        \opsize{\Strans(q_2)} - |w_2| & q_1 \not\in Q, q_2 \in Q \\
        \!\!
        \begin{array}{l}
        \min(\opsize{\Strans(q_1)} - |w_1|, \phantom{)} \\ % chktex 9
        \phantom{\min(} \opsize{\Strans(q_2)} - |w_2|)    % chktex 9
        \end{array}
        & \ensuremath{q_1, q_2 \in Q}
        \end{cases}
\]
\end{definition}

We can use leap size to define a notion of (symbolic) bisimilarity that can take larger steps; this will help us to formally justify the soundness of multi-step weakest preconditions.

\begin{definition}[Bisimulation with Leaps]
A \emph{bisimulation with leaps} is a relation $R \subseteq C \times C$, such that for all $c_1 \mathrel{R} c_2$, (1)~$c_1 \in F$ if and only if $c_2 \in F$, and (2)~$\delta^*(c_1, w) \mathrel{R} \delta^*(c_2, w)$ for all $w \in \bits^{\sharp(c_1, c_2)}$.
A \emph{symbolic bisimulation with leaps} is a formula $\phi$ such that $\sem{\phi}_\mathcal{L}$ is a bisimulation with leaps.
\end{definition}

Bisimulations with leaps can be more concise because they do not need to constrain configurations where both P4A are just buffering input, waiting for the next transition.

\begin{example}%
\label{ex:somewhat-fewer-pairs}
Recall the bisimulation from \cref{ex:lots-of-pairs}.
This relation contains the bisimulation with leaps $R'$, which is the smallest relation satisfying the rules
\begin{mathpar}
    \inferrule{%
        w \in \bits^{32}
    }{%
        \angl{\mathtt{q2}, s_1, w} \mathrel{R'} \angl{\mathtt{q5}, s_2, \epsilon}
    }
    \and
    \inferrule{%
        q \in \{ \acc, \rej \}
    }{%
        \angl{q, s_1, \epsilon} \mathrel{R'} \angl{q, s_2, \epsilon}
    }
\end{mathpar}
\end{example}

Bisimilarity with leaps is a sound and complete proof principle for language equivalence, which we record as follows.

\begin{restatable}{lemma}{restatebisimulationwithleapssoundcorrect}
Let $\phi$ be a formula.
The following are equivalent:
\begin{enumerate}
    \item
    There exists a symbolic bisim.\ with leaps $\psi$ s.t.\ $\phi \vDash \psi$.

    \item
    If $c_1 \sem{\phi}_\mathcal{L} c_2$, then $L(c_1) = L(c_2)$.
\end{enumerate}
\end{restatable}
% \begin{proof}[Proof sketch]
% For soundness, we show that bisimulation with leaps is a sound instance of the \emph{bisimulation up-to} proof principle~\cite{sangiorgi-1995}.
% For completeness, we apply \cref{lemma:symbolic-bisimulation-sound-complete} to obtain a symbolic bisimulation $\psi$ such that $\phi \vDash \psi$.
% We can then conclude by observing that any (symbolic) bisimulation is in particular a (symbolic) bisimulation with leaps.
% \end{proof}

\noindent
We can adapt \cref{algorithm:symbolic-refinement} to calculate the weakest symbolic bisimulation with leaps instead, if we adapt the axiomatization of the weakest precondition operator, as follows.

\begin{theorem}
\cref{algorithm:symbolic-refinement} remains correct if we change the condition on $\WP$ to require that, for all formulas $\phi$ and all $c_1, c_2 \in C$, the following equivalence holds:
\begin{mathpar}
    \forall w \in \bits^{\sharp(c_1, c_2)}.\ \delta^*(c_1, w) \sem{\phi}_\mathcal{L} \delta^*(c_2, w)  \vspace{-2mm} \\ {} \iff
    \forall \psi \in \WP(\phi).\ c_1 \sem{\psi}_\mathcal{L} c_2
\end{mathpar}
\end{theorem}
% \begin{proof}[Proof sketch]
% We can again refine the proof of \cref{theorem:symbolic-refinement-correct}.
% This time, the first loop invariant is that if $\rho$ is a symbolic bisimulation with leaps, then $\rho \vDash \bigwedge R \wedge \bigwedge T$; the third loop invariant states that configurations related by $\bigwedge R \wedge \bigwedge T$ transition into configurations related by $\bigwedge T$ by taking the number of steps dictated by the leap size function.
% \end{proof}

We can adapt the existing definition of $\WP$ to conform to this specification: simply repeat $\WP^<$ and $\WP^>$ as many times as is indicated by the source templates $t_1'$ and $t_2'$.
%Our Coq implementation is slightly more sophisticated, providing versions of $\WP^<$ and $\WP^>$ can take more than one step.

\subsection{Combining Optimizations}

The optimizations discussed are largely orthogonal.
However, their combination naturally gives rise to a third optimization, where $\reach_\phi$ is computed using leaps as well.
This results in an algorithm that computes a symbolic bisimulation with leaps that does not constrain intermediate (buffering) configurations.
We refer to the Coq development for full details.

%% file: implementation.tex
\begin{table}[t]
  \centering\small
  \vspace{17pt}
  \caption{Concepts from earlier in this paper and their realizations in the
  implementation.}%
  \label{table:impl}
  \begin{tabular}{@{}lll@{}}
    Paper name & Coq name & Implemented as \\
    \midrule
    $\Saut$ (\Cref{fig:internal-syntax}) &
    \texttt{Syntax.t} &
    Dependent record \\
    $\Sexp$ (\Cref{fig:internal-syntax}),
    \typexpbare{} &
    \texttt{expr} &
    Type-indexed ind. \\
    $\WP$ &
    \texttt{wp} &
    Gallina function \\
    $\bigwedge R \vDash \psi$ &
    \texttt{interp\_entailment} &
    Gallina function \\
    $\phi \vDash \bigwedge R$ &
    \texttt{interp\_entailment'} &
    Gallina function \\
    Bisimilarity &
    \texttt{bisimilar} &
    Coinductive relation \\
    \Cref{algorithm:symbolic-refinement} &
    \texttt{pre\_bisimulation} &
    Inductive relation \\
    {\normalfont\textbf{if}}
    $\bigwedge R \vDash \psi$ \dots
    &
    \texttt{decide\_entailment} &
    \LTac
  \end{tabular}
\end{table}

\section{Implementation}\label{sec:impl}
We implement \SYSTEM in Coq~\cite{bertot2013interactive,coq-web} using the Equations plugin~\cite{sozeau-2010,sozeau-mangin-2019}.
See \cref{table:impl} for a summary of how concepts from the formal development
in \Cref{sec:automata,sec:equiv,sec:optimizing} map to Coq notions.
Although an implementation in a different language might be more
efficient, our use of Coq produces rich semantic automata definitions
and reusable proofs of equivalence.
These artifacts are defined in Coq's expressive higher-order logic,
so they can be reused and composed with other mechanized logics
hosted in Coq like the Mathematical Components
library~\cite{mahboubi2017mathematical} or verification tools like the Verified Software
Toolchain~\cite{appel2011verified}.

\subsection{Automated Proof Search}
The most direct way to implement algorithms in Coq is by writing them
as functions in Gallina, Coq's functional programming language, but
unfortunately Gallina does not have I/O. As a consequence a Gallina
implementation of \cref{algorithm:symbolic-refinement} would have to
include a hand-written decision procedure for entailments $\bigwedge R
\vDash \psi$.
We instead realize \cref{algorithm:symbolic-refinement} in Coq as an
inductive relation (\cref{fig:pre-bisim-coq}), so we can rely on
external SMT solvers to handle entailments. 
This has the added benefit of sidestepping Coq's termination checker.\footnote{
  In particular, our pen-and-paper termination proof of \cref{algorithm:symbolic-refinement} does not directly translate to Coq's guarded primitive recursion~\cite{gimenez1994codifying}. 
}
The algorithm is run by performing proof search within the inductive relation,
and each step of the search proceeds by checking an entailment in the high-level automata logic.
While the logic of entailments is close to SMT's theory of bitvectors,
it also has richer terms that need to be desugared
(for example a finite-map encoding of the program store, constraints on the input packet length, constraints on the automata states, etc.).

\subsection{Reduction to SMT}
To reach a low-level logic amenable to off-the-shelf solvers,
we simplify formulas before checking them, through
a chain of verified simplifications and translations (\cref{fig:overview-architecture}).

This compilation turns formulas from the high-level logic \textsf{ConfRel} into
low-level first-order formulas over bitvectors, \textsf{FOL(BV)}. % chktex 36
In order, the implementation performs
\begin{inparaenum}[(1)]
  \item algebraic simplifications,
  \item template filtering,
  \item \textsf{FOL} compilation, and
  \item store elimination.
\end{inparaenum}
We now elaborate on each step.

\begin{figure}
\begin{lf-grey}[language=Coq, basicstyle=\footnotesize]
Inductive pre_bisimulation
  : conf_rel -> list conf_rel -> list conf_rel -> Prop :=
| Skip: forall phi R psi T,
    pre_bisimulation phi R T ->
    interp_entailment R psi ->
    pre_bisimulation phi R (psi :: T)
| Extend: forall phi R psi T,
    pre_bisimulation (psi :: R) (T ++ wp psi) ->
    ~ interp_entailment R psi ->
    pre_bisimulation R (psi :: T)
| Done: forall phi R,
    interp_entailment' phi R ->
    pre_bisimulation phi R [].
\end{lf-grey}
  \caption{\Cref{algorithm:symbolic-refinement} as an inductive relation in Coq.}%
  \label{fig:pre-bisim-coq}
\end{figure}

\begin{figure}
\begin{lf-grey}[language=Coq, basicstyle=\footnotesize]
Lemma small_filter_equiv:
  lang_equiv_state
    (P4A.interp IncrementalBits.aut)
    (P4A.interp BigBits.aut)
    IncrementalBits.Start
    BigBits.Parse.
Proof.
  solve_lang_equiv_state_axiom
    IncrementalBits.state_eqdec
    BigBits.state_eqdec
    false.
Time Qed.
\end{lf-grey}
  \caption{A Coq proof that the states \texttt{Start} and \texttt{Parse} of two
  automata named \texttt{IncrementalBits} and \texttt{BigBits} accept the same
  language (\texttt{lang\_equiv\_state}). The tactic
  \texttt{solve\_lang\_equiv\_state\_axiom} takes decision procedures for
  equality on the state sets of each automaton and a flag controlling a tactic optimization for large problems (here \texttt{false}).}%
  \label{fig:example-proof-coq}
\end{figure}

First, we use smart constructors to apply local algebraic simplifications.
Each application of the weakest precondition operator increases the size of a formula,
so these simplifications help prevent the formulas from growing too quickly.

Second, we perform template filtering to discard unused premises from
entailments.
Entailments have the form \[\bigwedge \phi_i \implies \psi_i \vDash \phi \implies
\psi,\] where $\phi$ and all $\phi_i$ are templates.
We discard any conjunct with $\phi_i \neq \phi$ and emit a simplified
entailment $\phi \vDash \bigwedge \psi_i \implies \psi$.
This puts our goal in the logic \textsf{ConfRelSimp}.

Third, we embed \textsf{ConfRelSimp} into the more general \textsf{FOL(Conf)} % chktex 36
syntax, removing references to states.
This fragment is the first-order theory of bitvectors and finite maps.

Finally, the store elimination pass fits formulas into the theory of
bitvectors \textsf{FOL(BV)}, by turning finite maps into first-order variables. % chktex 36
This is necessary because some SMT solvers we targeted do not support the
theory of finite maps.

\subsection{Querying Solvers}
The final \textsf{FOL(BV)} formula is serialized to SMT-LIB by a custom Coq % chktex 36
plugin and passed to an off-the-shelf SMT solver that can be selected using a custom vernacular command.
Currently, we support Z3~\cite{de2008z3}, CVC4~\cite{barrett-etal-2011}, and Boolector~\cite{niemetz-etal-2014}.

Before implementing our own plugin, we tried existing SMT integrations
for Coq, including CoqHammer~\cite{czajka-kaliszyk-2018} and
SMTCoq~\cite{armand-etal-2011}.
Neither solved our problem:
CoqHammer scaled poorly due to its flexible SMT encoding and proof
search procedure,
while SMTCoq performed better but lacked support for quantifiers.
Note however that, in contrast with our plugin, both of these tools perform \emph{proof reconstruction} to
produce a Coq proof term from solver output.
Consequently, our proof search must \emph{trust} the output of the SMT solver and our plugin.
A straightforward technique for doing this is to directly \texttt{admit} the low-level goals once the plugin has given the thumbs up.
This is rather error-prone because it means \texttt{admit} is used within automation,
and moreover,
it forces the final proof to be \texttt{Admitted} by the Coq kernel.
An alternative is to use a pair of \emph{axioms} for positive and negative validity of formulas in the low-level logic
and use the output of the SMT solver to conditionally apply the axioms.
While this approach is less performant because the Coq kernel checks the resulting term,
it allows for closed proof terms and avoids accidental misuse of \texttt{admit} in automation.

\subsection{Soundness and Trusted Computing Base}
The most important metatheoretic goal is to ensure that our algorithm produces
a certificate of equivalence only when the input parsers are indeed equivalent.
Towards this goal, our certificate-producing equivalence checker has a compact
TCB, with soundness relying on the Coq definitions of automata and automata
equivalence, the correctness of the SMT solver, the faithfulness of the
pretty-printing plugin, and the soundness of the Coq typechecker extended with
Streicher's axiom K~\cite{streicher-1993}.
The SMT solver and plugin (and the corresponding use of \texttt{admit}/axioms) are used
only in the proof search algorithm and could be removed from the TCB by implementing proof
reconstruction.

Our Coq development proves the soundness theorems stated in the paper,
but omits completeness and termination arguments.
Our proof search is really a semi-decision procedure: either the tactic
finds a proof and produces a Coq proof term, or it does not find a
proof and no certificate is produced.
Trustworthy certificates, our main metatheoretic goal, only require
a mechanization of soundness.
In fact, the only termination or completeness bug we encountered arose
from incorrectly interpreting failed SMT queries as UNSAT, which was a
bug in the plugin and not in the algorithm itself.

\begin{figure}
  \centering
  \includegraphics[clip=true,trim=510pt 136pt 660pt 96pt,width=0.8\columnwidth,page=6]{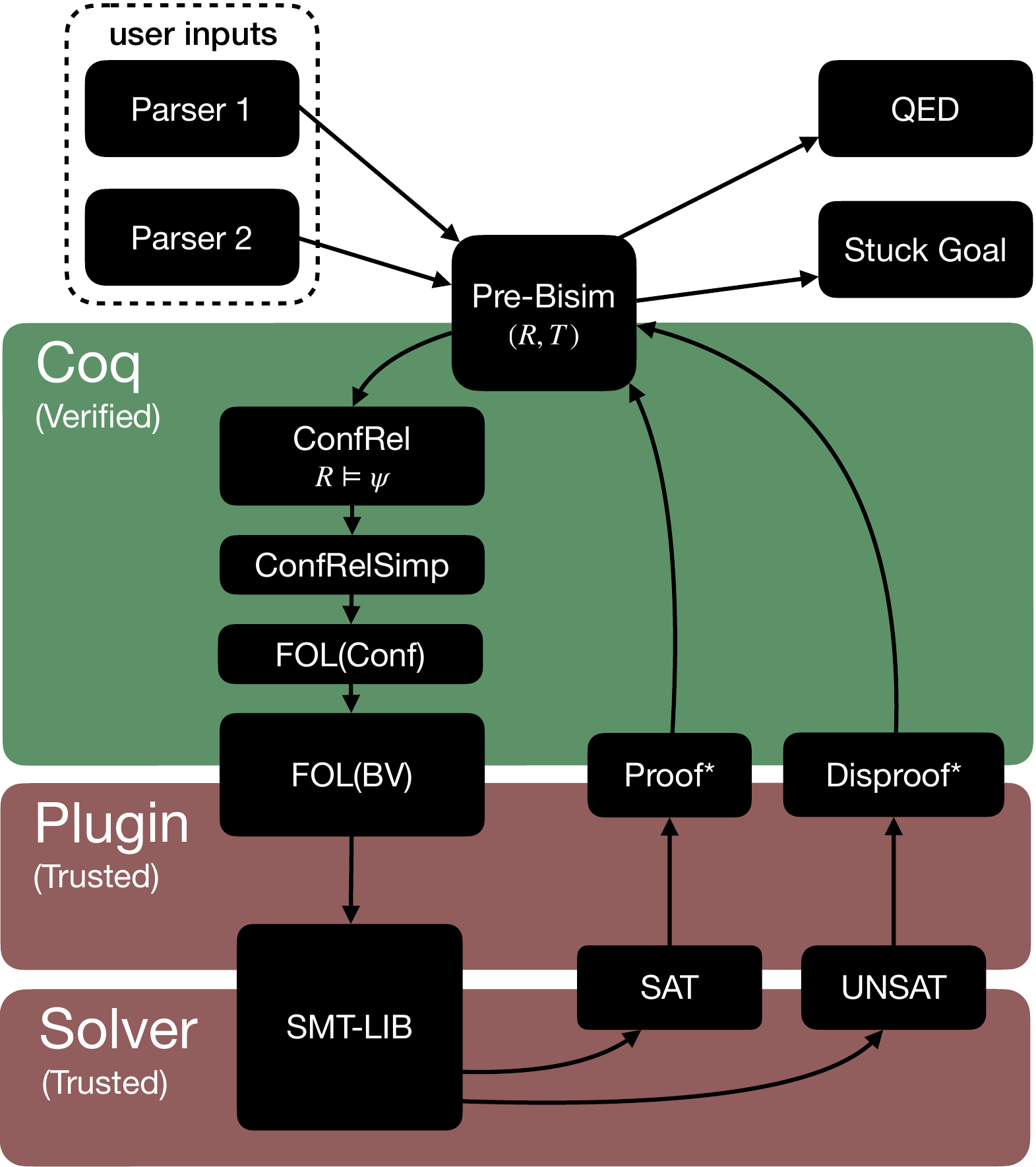}
  \caption{The \SYSTEM implementation architecture.
  In each iteration, a \textsf{ConfRel} formula is checked by
  reduction to SMT via a chain of intermediate logics (at left).
  A Coq plugin pretty-prints \textsf{FOL(BV)} syntax to SMT-LIB syntax % chktex 36
  and invokes the SMT solver.
  The asterisk (*) on \textsf{Proof} and \textsf{Disproof} (at right) indicates
  that the plugin does not produce proofs.
  When the procedure halts, either the Coq goal is provable
  (\textsf{QED}), or the goal is stuck and no certificate is
  produced.}%
  \label{fig:overview-architecture}
\end{figure}

%% file: case-studies.tex
\begin{table*}
  \vspace{17pt}
  \centering
  \caption{Parsers in our evaluation: \textbf{States} gives the total number of states in both parsers,
  \textbf{B}ranched gives the number of bits in automata transition select statements,
  \textbf{T}otal gives the number of bits across all variables, and \textbf{Runtime} and \textbf{Memory} give the aggregate runtime and maximum resident size. An optimal
  verification algorithm would need to represent $2^\textbf{B}$ states, while an
  explicit state space would contain $2^\textbf{T}$ states. An asterisk on the memory use indicates an out-of-memory exception.
  }\label{tab:benchmarks}
  \begin{small}
  \begin{tabular}{llrrrrr}\toprule
    & \textbf{Name} & \textbf{States} & \textbf{B}ranched (bits) & \textbf{T}otal (bits) & \textbf{Runtime} (minutes) & \textbf{Memory} (GB) \\\midrule
    \multirow{5}{*}{\STAB{\rotatebox[origin=c]{90}{\textbf{Utility}}}} &  State Rearrangement & 5 & 8 & 136 & 0.12 & 0.66  \\
    & Variable-length parsing & 30 & 64 & 632 & 953.42 & 405.64 \\
    & Header initialization & 10 & 10 & 320 & 15.95 & 13.71 \\
    & Speculative loop & 5 & 2 & 160 & 4.12 & 3.16  \\
    & Relational verification & 6 & 64 & 1056  & 1.68 & 2.07 \\
    & External filtering & 6 & 64 & 1056 & 1.18 & 1.71 \\\midrule
    \multirow{5}{*}{\STAB{\rotatebox[origin=c]{90}{\textbf{Applicability}}}} & Edge & 28 & 52 & 3184 & 528.38 & 251.26 \\
    & Service Provider & 22 & 50 & 2536 & 1244.5 & $499.80^*$ \\
    & Datacenter & 30 & 242 & 2944 & 1387.95 & 404.50 \\
    & Enterprise & 22 & 176 & 2144 & 217.93 & 66.13 \\
    & Translation Validation & 30 & 56 & 3148 & 746.2 & 350.48 \\
  \bottomrule
  \end{tabular}
  \end{small}
\end{table*}
% We discuss each challenge and microbenchmark in turn.

\section{Evaluation}\label{sec:case-studies}

We evaluate \SYSTEM through case studies (listed in
\autoref{tab:benchmarks}) along two dimensions:
\begin{inparaenum}[(1)]
  \item the \emph{utility} of bisimulations for solving problems of interest in networking (and, by extension, the \emph{expressiveness} of \SYSTEM), and
  \item the \emph{applicability} of \SYSTEM to real-world parsers (and, by extension, its \emph{scalability} to non-trivial inputs).
\end{inparaenum}

\subsection{Utility}\label{subsec:utility-studies}
To evaluate whether equivalence checks are useful in the networking
domain, we identified \emph{six} distinct verification tasks, and
showed how they can be solved with \SYSTEM. % chktex 13

\paragraph{State Rearrangement.}
Because parser states translate to hardware resources,
it is common for compilers to merge and split parser states,
to optimize the write and branch behavior for the particular hardware.
We implemented a reference parser for a stylized IP and UDP/TCP protocol
in which the prefix is 64 bits of IP
and the suffix is either 32 bits of UDP or 64 bits of TCP\@ (\autoref{fig:combined-header}).
Note that the TCP and UDP headers share a common prefix of 32 bits. We
then implemented an optimized parser that extracts the IP and common
prefix, and then branches to determine how to parse the remaining
bits. We used \SYSTEM to show that the parsers accept the same
packets, even though they do so in different ways.

\paragraph{Variable-Length Formats.}
Handling formats with variable lengths, such as type-length-value (TLV)
encodings, is a common challenge in protocol parsing, because the
amount of data parsed in each state depends on a previously-parsed
values.
We implemented a parser for Internet Protocol options~\cite{ip-params},
a common variable-length networking format. Our parser handles up to
two generic options, with data-dependent lengths that range from 0
bytes to 6 bytes.
We also implemented a custom parser for the Timestamp option, in which
a specialized parser extracts the fields specific to its format. Again,
we used \SYSTEM to show that the parsers accept the same packets, even
though the header formats are variable and they do so in different
ways.

\paragraph{Header Initialization.}
A common error in P4 programs is reading from uninitialized headers.
In parsers, this can happen when several paths converge on a common
state, and the programmer has forgotten to write to a given header on
one or more of the paths.
For example, VLAN tags~\cite{vlanieee} are an optional 4-byte format
that can appear at the end of an Ethernet frame.
If the VLAN tag is present, its value can be used to influence routing
behavior. However, a common bug is to accidentally branch on an
uninitialized VLAN tag when it was not present in the packet.
To fix this bug, one can assign a default value to missing VLAN tags.
We implemented a parser for Ethernet, optional VLAN, IP, and UDP, that
either parses a VLAN tag or fills it with a default value if it is
missing (\autoref{fig:self-comp} of the appendix).
We used \SYSTEM to check that the set of accepted packets is
independent of the initial store. This check succeeds, so we conclude that the
parser \emph{not} depend on uninitialized headers.

\paragraph{Speculative Extraction.}
Many high-performance protocol parsers \emph{speculatively extract}
packet data and then make control-flow decisions based off the
contents of that data.
We implemented the example from \cref{fig:example-parser-syntax} with MPLS followed
by UDP, in which the body of the optimized MPLS loop speculatively
extracts two MPLS headers. If the first of these indicates the end of
the header, then the parser has overshot the MPLS header, and the
remaining data must be reinterpreted as a UDP packet. We used \SYSTEM
to verify that these parsers accept the same packets.

\begin{figure}
  \centering
  \def\arraystretch{0}
  \begin{tabular}{p{0.225\textwidth}:p{0.22\textwidth}}
    \begin{lf-grey}[basicstyle=\footnotesize\ttfamily]
parse_ip {
  extract(ip, 64);
  select(ip[40:43]) {
    (0001) => 
      parse_udp
    (0000) => 
      parse_tcp
  }
}
parse_udp {
  extract(udp, 32);
  goto accept
}
parse_tcp {
  extract(tcp, 64);
  goto accept
}
\end{lf-grey}
  &
    \begin{lf-grey}[basicstyle=\footnotesize\ttfamily]
parse_combined {
  extract(ip, 64);
  extract(pref, 32)
  select(ip[40:43]) {
    (0001) => 
      accept
    (0000) => 
      parse_suff
  }
}
parse_suff {
  extract(suff, 32);
  goto accept
}

\end{lf-grey}
  \end{tabular}
  \caption{Reference and combined parsers for a stylized IP and TCP/UDP protocol.}%
  \label{fig:combined-header}
\end{figure}

\begin{figure*}[ht]
  \centering
  \includegraphics[clip=true,trim=0pt 270pt 0 0,width=0.85\textwidth]{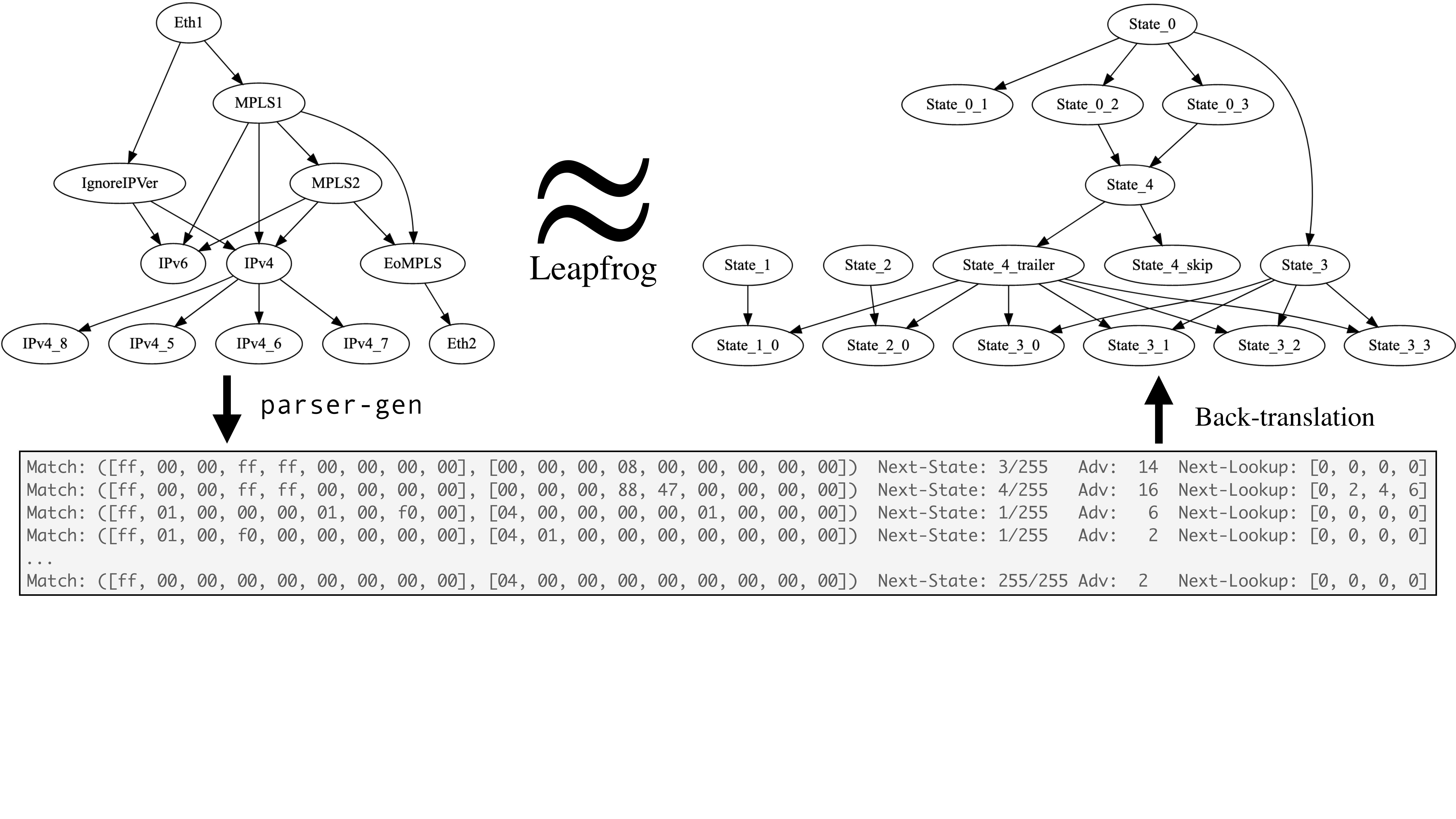}
  \caption{The Edge stack case study. The original parser (left) is compiled to a table (below,
  most entries elided), which we translate back into a parser (right) and
  prove equivalent to the original.}%
  \label{fig:trans-valid}
  \vspace{0.2em}
\end{figure*}

\paragraph{External Filtering.}

Another common idiom is to implement a lenient parser that accepts
well-formed and malformed packets, and then compensate with an
external filter (e.g., by dropping packets later).
Recall that the last two bytes of an Ethernet header are a type field that
determines the header that follows---e.g., IPv4, IPv6, or something
else.
We implemented two parsers: a lenient parser that assumes the input
packet is IPv6 if it is not IPv4, and a strict parser that explicitly
checks the Ethernet type field and rejects other types of packets.
We modeled an external filter for the lenient parser by picking an
initial relation that not only requires the states to be equally
accepting, but to also terminate with a store where the Ethernet type
is IPv4 or IPv6.
We then computed a bisimulation modulo this initial relation and prove
that the two parsers are equivalent.
This case study shows that \SYSTEM can do more than just relate the
sets of accepted packets: it can also relate the values in the stores.

\paragraph{Relational Verification.}
\SYSTEM can also verify other useful relational properties of parsers.
For instance, consider two parsers that extract data into differently named
headers (e.g., the example from \Cref{sec:overview}), or even with
different fragments of the input packet scattered across the store.
\SYSTEM can be used to phrase and verify relations between
parser stores. To demonstrate this, we verified that when the lenient
and strict parsers from the previous case both accept some packet,
there is a correspondence between the values in their stores.
We picked an initial relation that requires the values for headers on
the left to correspond to the values for headers on the right,
provided both configurations are accepting and used \SYSTEM to establish
a pre-bisimulation.
Compared to language equivalence, these applications do not have as much
metatheory developed in Coq, where there is a lemma connecting an appropriately
configured pre-bisimulation to language equivalence.
However, we believe our technique is sound and could be justified in Coq.

\medskip
Separately from these six tasks, we used \SYSTEM to compare parsers that did \emph{not} accept the same packets, such as the two parsers of the \emph{external filtering} task.
This was done as a sanity check to see if (1)~the proof script still terminated and (2)~it did not erroneously claim to prove equivalence.
A failure would indicate a bug in our pen-and-paper analysis of the algorithm, or our trusted codebase.
Fortunately, \SYSTEM acts as expected, by reaching the end of the main loop, then failing when trying to apply the \texttt{Close} step.

\subsection{Applicability}

To evaluate \SYSTEM's applicability to real-world parsers, we encoded
the benchmarks used by the developers of the \texttt{parser-gen}
tool~\cite{gibb2013design}. It provides parsers for \emph{four}
different scenarios:
\begin{inparaenum}[(1)]
  \item Edge, for a gateway router,
  \item Service Provider, for a core router,
  \item Datacenter, for a top-of-rack switch in a cloud,
  \item and Enterprise, for a router in a campus or company network.
\end{inparaenum}
Each of these parsers supports a different set of protocols depending
on its intended use.%
\footnote{We did not consider one of the parsers discussed in the
  \texttt{parser-gen} paper, Big-Union, which models the combined
  features from all four scenarios. Unlike the others, Big-Union does
  not model a typical scenario but is primarily intended for bounding
  hardware requirements.}
We translated each of these parsers into a corresponding P4A parser
and used \SYSTEM to perform a \emph{self-comparison} check---i.e., we
verified that each parser is equivalent to itself.

Next, we used \SYSTEM to perform translation validation. The
\texttt{parser-gen} framework also comes with a compiler that takes a
parse graph (analogous to a P4A) and compiles it to an
efficient hardware representation.
The compiler models constraints at the hardware level (e.g.,
limiting the number of bits that can be extracted or branched on in
each state) and incorporates sophisticated optimizations to make the
best use of limited resources (e.g., splitting and merging states).

We ran the \texttt{parser-gen} compiler on the parser for the Edge
router, which generated a hardware-level representation with states,
instructions, and transitions encoded in a
table---see~\Cref{fig:trans-valid}.
We then wrote a script to translate the table representation back into a P4 automaton.%
\footnote{
  While the two languages are similar,
  the parser-gen hardware representation is different enough from P4A (mainly due to unproductive states and speculative lookahead transitions) to make the reverse translation fuzzy.
  Of all of the parser-gen benchmarks, we found that Edge's hardware table was the closest to P4A and required the least amount of manual repair.
  This technique could in principle be adapted to other parser-gen benchmarks;
  while they are a bit larger and could stress \SYSTEM's scaling,
  the main challenge is a robust translation from hardware tables to P4A.
  }
Finally, we used \SYSTEM to check the equivalence of the
two parsers.

We were able to prove that the \texttt{parser-gen} compiler preserves
the semantics of the original Edge P4A automata. Hence, \SYSTEM was able to validate a
third-party compiler's output on its own benchmark program. Note that we 
designed \SYSTEM before we had experience using \texttt{parser-gen}.

\subsection{Discussion}

Overall we find that \SYSTEM can be applied to a diverse set of
practical scenarios. In the rest of this section, we discuss some of
our experiences using the tool, including its limitations and
directions for future work.

\paragraph{Automation}
In the early stages of this work, we derived and validated the relevant
bisimulations without automated tactics.
This turned out to be a significant proof burden---e.g., our manual
equivalence proof for the State Rearrangement case study took two
weeks of work.
In contrast, the push-button \SYSTEM proof takes only six seconds on a laptop.
Although \SYSTEM could be adapted to be more interactive,
letting the user apply \texttt{Skip} or \texttt{Extend} and prove the required entailment,
we believe that its power lies in the convenience offered by delegating goals to an SMT solver.

\SYSTEM is particularly useful in situations where it is difficult to see
whether two parsers are equivalent, such as in the translation validation
experiment.
While we spent a few days trying to prove the translation validation parsers equivalent on pen-and-paper, 
we were unsure that they were actually equivalent until the \SYSTEM proof succeeded.

\paragraph{SMT Solver Performance}
SMT solvers have unpredictable performance.
We used Z3 for the queries in most of our benchmarks,
but sometimes needed to switch to CVC4.
Overall we found that all of the queries were solved in at most 10 seconds,
with~99\% taking at most 5 seconds.
It was easy to switch between SMT solvers because we targeted a well-supported
subset of the SMT-LIB query format (namely the theory of bitvectors).

\paragraph{Overall Performance}
Like any verification tool, \SYSTEM has limitations.
Scaling to large parsers is challenging due to the combinatorial
explosion of configurations.
All of the smaller experiments (up to around 10 states) were interactive on stock hardware,
finishing in $\leq 5$ minutes and $\leq 16$~GB of memory.
Most took several minutes.
For larger experiments, the larger state space lead to significantly
higher memory demands.
Coq needed 400 GB of RAM to verify the largest Applicability study
(Datacenter) and ran out of memory on the Service Provider study.
Although this is a lot, it's unsurprising
because the concrete state space for the Applicability study would have around $2^{242}$ elements.

The optimizations discussed in \Cref{sec:optimizing} had a significant impact.
Specifically, our smallest State Rearrangement benchmark went from 30 seconds and 1.7~GB of memory to
42 minutes and 36 GB of memory when leaps were disabled;
it did not finish without reachable state pruning.

\paragraph{Future Work}
One way to improve the scalability of \SYSTEM in the future could be
to investigate compositional reasoning techniques. Such techniques
could facilitate divide-and-conquer strategies, allowing \SYSTEM to be applied
to larger parsers than our current implementation supports.

Another possibility is to vary the underlying algorithm.
One could imagine a symbolic treatment of Hopcroft and Karp's algorithm~\cite{hopcroft-karp-1971}, which approximates a suitable bisimulation from below, or Paige and Tarjan's \emph{partition refinement} algorithm~\cite{paige-tarjan-1987}, which represents the current approximation of the largest bisimulation in terms of its equivalence classes.
For the latter, one would need to estimate the number of configurations in a symbolically represented equivalence class to choose the next block to split.

The bulk of \SYSTEM's memory usage is occupied by the proof object generated by our \LTac{} script.
Alternatively, one could implement the same algorithm in Gallina, axiomatizing the decision procedure, and extract it to OCaml.
While such an approach is likely to be more efficient, it would also undermine our goal of producing a proof object that is reusable in a larger verification effort.

P4As are an abstraction of P4 parsers.
For one thing, they do not incorporate externs, which are architecture-specific extensions that support, for instance, checksum algorithms or persistent state.
In addition, P4 parsers support arrays (in the form of header stacks), subparser calls, and parser lookahead, all of which are not part of our definition of P4 automata.
More work is necessary to see whether P4As can be extended to support or simulate these features.

In the future, we would like to use \SYSTEM's equivalence
checks to systematically perform translation validation on other
networking stacks. For example, one could imagine writing a library of
reference implementations for protocols defined in RFCs, and checking
that real-world implementations conform to those standards.

%% file: related.tex
\section{Related work}\label{sec:related}

\paragraph{Automata Equivalence Checking.}

Our algorithm is a variation on Moore's classical algorithm to decide all-pairs language equivalence in a DFA~\cite{moore-1956}.
Moore's approach was later improved upon by \emph{partition refinement}~\cite{hopcroft-1971,kanellakis-smolka-1983,paige-tarjan-1987}.
We deviate from these classical procedures in two key aspects.

First, instead of using concrete data structures we use symbolic ones.
This idea goes back to Coudert et al.~\cite{coudert-etal-1989}, and has since been widely applied~\cite{bouajjani-etal-1990,bouajjani-etal-1992,burch-etal-1992,bouali-desimone-1992,derisavi-2007}.
These algorithms use Binary Decision Diagrams (BDDs) as their symbolic representation.
Other authors favored a logical representation, combined with decision procedures for the logic~\cite{hennessy-lin-1995,feng-etal-2013}.
Dehnert et al.~\cite{dehnert-etal-2013} makes use of an SMT solver to decide questions about the logical representations.

Second, instead of maintaining a list of equivalence classes, we
maintain a representation of an equivalence relation. The earliest
instance of this we have been able to track down is due to Bouali and
De Simone~\cite{bouali-desimone-1992}. Mumme and Ciardo observed that
such an approach is particularly beneficial when there tend to be a
large number of equivalence classes~\cite{mumme-ciardo-2011}.

Algorithms based on \emph{bisimulation up to
  congruence}~\cite{bonchi-pous-2013,dantoni-etal-2018} are similar in
the sense that they mitigate state space explosion---in this case, as
a result of determinization. They exploit the internal
structure of the expanded state space to terminate early, something
that inspired us propose the notion of a bisimulation with leaps.

\paragraph{Network Verification.}
The \texttt{p4v} verifier~\cite{p4v}, the verifier of Neves et
al.~\cite{neves-p4}, and Aquila~\cite{aquila} are push-button
verifiers for functional properties of P4 programs, including P4
parsers.
They work by translating to a verification IR (either guarded command
language~\cite{dijkstra-1975} or simple C) and then analyzing the IR\@.
None of these tools produce proofs, and their translations are not
proved sound with respect to a reusable semantics of P4.
Moreover, these tools verify \emph{functional} specifications about a single P4 program.
Our work is complementary because by contrast, our tool produces \emph{relational} proofs grounded in a
reusable Coq semantics for two P4 automata. 
Aquila includes a self-validation system for finding semantic bugs in
the verifier.
Defining our semantics in Coq allowed us to foundationally
\emph{prove} the absence of semantic bugs, so while \SYSTEM does not
need self-validation, it could be an oracle for validating other tools.

The Gauntlet translation validator checks program equivalence for P4 programs
without parsers or externs.
We see this work as complementary to Leapfrog, which focuses on parser
equivalence.
Outside the parser, P4 programs have loop-free control flow, complex data
structures, and rich semantic actions.
Inside the parser, P4 programs have loops, simpler data structures, and simpler
semantic actions.
Consequently, parser verification is concerned with control flow more than
anything else, making it a different kind of verification problem than
verification for the rest of a P4 program.

\paragraph{Automatic Foundational Verification}
SpaceSearch~\cite{spacesearch} exposes a high-level solver interface to search
large state spaces. In contrast, our solver interface is lower level, and our
tool avoids extraction to produce a Coq certificate.

CreLLVM~\cite{crellvm} instruments LLVM to produce translation
validation proofs in a relational Hoare logic, resulting in a compact
TCB and reusable Coq proof certificate. \SYSTEM has a similar TCB, but
completeness (\cref{theorem:symbolic-refinement-correct}) means it
does not require proof hints.

The Narcissus~\cite{narcissus} and EverParse~\cite{everparse} tools
synthesize correct parsers and serializers from high level
descriptions of packet formats using verified parser combinator
libraries.
Synthesis and equivalence are related but distinct problems, and our
tool is complementary to synthesis tools.
For instance, a P4 parser generated by a parser synthesizer like
EverParse might be further optimized by a P4 compiler to run on
hardware.
\SYSTEM could validate the results of compilation, preserving the
guarantee offered by the synthesizer.

GPaco~\cite{hur2013power, hur2020equational} is a framework for
modular coinductive reasoning in Coq, which supports ``up-to''
bisimilarity techniques.
It is designed for interactive use and focuses on automating low-level
proof steps.
GPaco may be useful for generalizing our mechanized metatheory for
leaps.

%% file: acks.tex
%% Acknowledgments
\begin{acks}                            %% acks environment is optional
  %% contents suppressed with 'anonymous'
%% Commands \grantsponsor{<sponsorID>}{<name>}{<url>} and
%% \grantnum[<url>]{<sponsorID>}{<number>} should be used to
%% acknowledge financial support and will be used by metadata
%% extraction tools.
  We thank Glen Gibb for help understanding and using his \texttt{parser-gen} framework.
  We received helpful feedback on the writing from
  Glen Gibb,
  James R. Wilcox,
  Bill Harris,
  and members of the Cornell programming languages group;
  we thank them for their feedback.

  R.~Doenges and N.~Foster were supported in part by the National Science
  Foundation under grant FMiTF-1918396, DARPA under contract HR0011-20-C-0107,
  and gifts from Fujitsu, Google, InfoSys, and Keysight.

T.~Kapp\'{e} was partially supported by the European Union’s Horizon 2020 research and innovation programme under the Marie Sk\l{}odowska-Curie grant agreement No. 101027412 (VERLAN).
J.~Sarracino and G.~Morrisett were supported by DARPA contract HR0011-19-C-0073.
\end{acks}

%% file: appendix.tex
\newpage

\section{Omitted case study figures}

See \Cref{fig:self-comp}, \Cref{fig:self-comp2},
\Cref{fig:self-comp3-1}, and \Cref{fig:self-comp3-2} for the definitions of 
additional P4 automata used in the case studies.

\begin{figure}[h]
  \centering
  \small
  \begin{tabular}{c:c}
  \begin{subfigure}{0.21\textwidth}
  \begin{lf-grey}
 parse_eth {
   extract(ether, 112);
   select(ether[0:0]) {
     0 => default_vlan
     1 => parse_vlan
   }
 }
 default_vlan {
   vlan := 0x0000;
   extract(ip, 160)
   goto parse_udp
 }
 parse_vlan {
   extract(vlan, 32);
   goto parse_ip
 }
\end{lf-grey}
  \end{subfigure}
  &
  \begin{subfigure}{0.22\textwidth}
  \begin{lf-grey}
 parse_ip {
   extract(ip, 160);
   goto parse_udp
 }
 parse_udp {
   extract(udp, 64);
   select(vlan[0:3]) {
     1111 => reject
     _    => accept
   }
 }
\end{lf-grey}
\vspace*{8.5em}\null
\end{subfigure}
  \end{tabular}
  \caption{Ethernet stack parser with an optional VLAN tag.}%
  \label{fig:self-comp}
\end{figure}

\begin{figure}[ht]
\vspace{1em}
\centering
\small
\begin{tabular}{c:c}
\begin{subfigure}{0.457\columnwidth}
\begin{lf-grey}
parse_eth {
  extract(ether, 112);
  select(ether[96:111]) {
    0x86dd => parse_ipv6
    0x8600 => parse_ipv4
  }
}
parse_ipv6 {
  extract(ipv4, 288)
  goto accept
}
parse_ipv4 {
  extract(ipv6, 128);
  goto accept
}
\end{lf-grey}
\vspace*{.5em}\null
\end{subfigure}
&
\begin{subfigure}{0.462\columnwidth}
\begin{lf-grey}
parse_eth {
  extract(ether, 112);
  select(ether[96:111]) {
    0x86dd => parse_ipv6
    0x8600 => parse_ipv4
    _   => reject
  }
}
parse_ipv6 {
  extract(ipv4, 288)
  goto accept
}
parse_ipv4 {
  extract(ipv6, 128);
  goto accept
}
\end{lf-grey}
\end{subfigure}
\end{tabular}
\caption{Sloppy and strict parsers for Ethernet and IP\@.}%
\label{fig:self-comp2}
\end{figure}

\section{Proofs omitted from Section~\ref{sec:equiv}}

\restatesymbolicbisimulationsoundcomplete*
\begin{proof}
(1) implies (2) by definition of symbolic bisimulation.
The proof that (2) implies (3) is a standard argument from automata theory, showing that for all $c_1 \mathrel{R} c_2$ and $w \in \bits^*$, it holds that $w \in L(c_1)$ if and only if $w \in L(c_2)$, by induction on $w$.
Finally, to see that (3) implies (1), we construct $\psi$ to relate all configurations with equivalent languages, i.e.,
\begin{mathpar}
    \psi := \bigvee_{L(c_1) = L(c_2)} \phi^<(c_1) \wedge \phi^>(c_2)
\quad \text{where for $\lessgtr \in \{ <, > \}$, } \\
\phi^\lessgtr(q, s, w) := \Sstatepred{q}{\lessgtr} \wedge \bigwedge_{h \in H} \Sbhdr{\lessgtr} = s(h) \wedge \Sbbuf{\lessgtr} = w
\end{mathpar}
Clearly, ${\sem{\phi^<(c_1) \wedge \phi^>(c_2)}_\mathcal{L}} = \{ \angl{c_1, c_2} \}$ for all $c_1, c_2 \in C$, and hence $c_1 \sem{\psi}_\mathcal{L} c_2$ if and only if $L(c_1) = L(c_2)$.
This makes $\sem{\psi}_\mathcal{L}$ a bisimulation: configurations with the same
  language are equally accepting, and step to configurations with the same language; hence, $\psi$ is a symbolic bisimulation.
Since states related by $\sem{\phi}_\mathcal{L}$ agree on their languages, $\phi \vDash \psi$.
\end{proof}

\restatesymbolicrefinementcorrect*
\begin{proof}
Note that each iteration either the relation $\sem{\bigwedge R}_\mathcal{L}$ shrinks (if we enter the $\textbf{if}$-block), or $\sem{\bigwedge R}_\mathcal{L}$ stays the same but $T$ shrinks.
Thus, the main loop must terminate.

For partial correctness, we note that the loop maintains the following three loop invariants:
\begin{enumerate}
    \item
    If $c_1 \sem{\bigwedge R \wedge \bigwedge T}_\mathcal{L} c_2$, then for all $b \in \bits$ it holds that $\delta(c_1, b) \sem{\bigwedge R}_\mathcal{L} \delta(c_2, b)$.

    This claim holds trivially before the loop, because $R$ is empty.
    To see that it is preserved, let $T'$ and $R'$ be the new values of $T$ and $R$.
    There are two cases.
    \begin{itemize}
        \item
        If $\bigwedge R \vDash \psi$ then $R' = R$ and $T' \cup \{ \psi\} = T$.
        Thus,
        \begin{align}%
            {\sem{\bigwedge R' \wedge \bigwedge T'}_\mathcal{L}}
                &= {\sem{\bigwedge R \wedge \bigwedge T'}_\mathcal{L}} \tag*{} \\
                &= {\sem{\psi \wedge \bigwedge R \wedge \bigwedge T'}_\mathcal{L}} \tag*{} \\
                &= {\sem{\bigwedge R \wedge \bigwedge T}_\mathcal{L}}
            \label{equation:nochange}
        \end{align}
        Since $R = R'$ and the property was true before the loop, it also holds afterwards.

        \item
        Otherwise, $R' = R \cup \{ \psi \}$ and $T' = T \cup \WP(\psi)$.
        Suppose $c_1 \sem{\bigwedge R' \wedge \bigwedge T'}_\mathcal{L} c_2$; we should show that, for all $b \in \bits$, we have $\delta(c_1, b) \sem{\bigwedge R}_\mathcal{L} \delta(c_2, b)$ as well as $\delta(c_1, b) \sem{\psi}_\mathcal{L} \delta(c_2, b)$.
        Since $R' \cup T' \subseteq R \cup T$, we have in particular that $c_1 \sem{\bigwedge R \wedge \bigwedge T}_\mathcal{L} c_2$, and hence $\delta(c_1, b) \sem{\bigwedge R}_\mathcal{L} \delta(c_2, b)$, because the property is true before the loop.

        Furthermore, since $\WP(\psi) \subseteq R' \cup T'$, we also have $c_1 \sem{\chi}_\mathcal{L} c_2$ for all $\chi \in \WP(\psi)$.
        By the precondition about $\WP$, we conclude that $\delta(c_1, b) \sem{\psi}_\mathcal{L} \delta(c_2, b)$.
    \end{itemize}

    \item
    If $c_1 \sem{\bigwedge R \wedge \bigwedge T}_\mathcal{L} c_2$, then $c_1 \in F$ iff $c_2 \in F$.

    This property holds by construction before the loop.
    To see that it is preserved, note that $\sem{\bigwedge R \wedge \bigwedge T}_\mathcal{L}$ stays the same in the iterations where $\bigwedge R \vDash \psi$, and shrinks in all other iterations.
    Because the claim holds before each iteration, it must also hold afterwards.

    \item
    If $\rho$ is a symbolic bisimulation, then $\rho \vDash \bigwedge R \wedge \bigwedge T$.

    This claim holds before the loop, where $R = \emptyset$ and $T = I$.
    After all, if $c_1 \sem{\rho}_\mathcal{L} c_2$, then $c_1 \in F$ if and only if $c_2 \in F$, and hence $c_1 \sem{\bigwedge I}_\mathcal{L} c_2$.

    For preservation, let $T'$ and $R'$ be the new values of $T$ and $R$.
    There are again two cases.
    \begin{itemize}
        \item
        If $\bigwedge R \vDash \psi$, then by the derivation in~\eqref{equation:nochange}, we know that ${\sem{\bigwedge R' \wedge \bigwedge T'}_\mathcal{L}} = {\sem{\bigwedge R \wedge \bigwedge T}_\mathcal{L}}$.
        Since $\rho \vDash \bigwedge R \wedge \bigwedge T$, it then follows that $\rho \vDash \bigwedge R' \wedge \bigwedge T'$.

        \item
        Otherwise, we have that $R' \cup T' = R \cup T  \cup \WP(\psi)$.
        We should show that if $c_1 \sem{\rho}_\mathcal{L} c_2$, then $c_1 \sem{\chi}_\mathcal{L} c_2$ for all $\chi \in R' \cup T'$.
        Because the property holds before the loop, we already know that, under these circumstances, $c_1 \sem{\chi}_\mathcal{L} c_2$ for all $\rho \in R \cup T$; thus, it suffices to prove $c_1 \sem{\chi}_\mathcal{L} c_2$ for all $\chi \in \WP(\psi)$.
        By the precondition on $\WP$, it suffices to verify that for all $b \in \bits$ we have $\delta(c_1, b) \sem{\psi}_\mathcal{L} \delta(c_2, b)$.
        Now, since $\rho$ is a symbolic bisimulation, we already know that for all $b \in \bits$ we have $\delta(c_1, b) \sem{\rho}_\mathcal{L} \delta(c_2, b)$.
        Lastly, since the claim held before the loop, and $\psi \in T$, it then follows that $\delta(c_1, b) \sem{\psi}_\mathcal{L} \delta(c_2, b)$.
    \end{itemize}
\end{enumerate}
When the loop terminates, all of these conditions are true, and we also have that $T = \emptyset$.
The first two conditions then tell us that $\bigwedge R$ is a symbolic bisimulation; moreover, the second condition says that if $\rho$ is a symbolic bisimulation, then $\rho \vDash \bigwedge R$.
It then follows that $\bigwedge R$ is in fact the \emph{weakest} symbolic bisimulation.
This allows us to wrap up the correctness argument as follows:
\begin{itemize}
    \item
    If the algorithm returns $\True$, we know that $\phi$ entails a symbolic bisimulation, which by \cref{lemma:symbolic-bisimulation-sound-complete} tells us that $c_1 \sem{\phi}_\mathcal{L} c_2$ implies $L(c_1) = L(c_2)$.

    \item
    Conversely, if $c_1 \sem{\phi}_\mathcal{L} c_2$ implies $L(c_1) = L(c_2)$, then by \cref{lemma:symbolic-bisimulation-sound-complete} there exists a symbolic bisimulation $\psi$ such that $\phi \vDash \psi$.
    Since $\bigwedge R$ is the \emph{weakest} symbolic bisimulation, we also know that $\psi \vDash \bigwedge R$, and thus $\phi \vDash \bigwedge R$, meaning the algorithm returns $\True$.
    \qedhere
\end{itemize}
\end{proof}

\section{Proofs omitted from Section~\ref{sec:optimizing}}

\restatesymbolicrefinementinitializecorrect*
\begin{proof}
The same termination argument still applies.
The loop invariants are as follows.
\begin{enumerate}
    \item
    If $c_1 \sem{\bigwedge R \wedge \bigwedge T}_\mathcal{L} c_2$, then for all $b \in \bits$ it holds that $\delta(c_1, b) \sem{\bigwedge R}_\mathcal{L} \delta(c_2, b)$.

    The proof is completely analogous to the corresponding loop invariant in \cref{theorem:symbolic-refinement-correct}.

    \item
    If $c_1 \sem{\bigwedge R \wedge \bigwedge T}_\mathcal{L} c_2$ and $\floor{c_1} \mathrel{\reach_\phi} \floor{c_2}$, then $c_1 \in F$ if and only if $c_2 \in F$.

    Note that this invariant is slightly weaker than the corresponding invariant in the proof of \cref{theorem:symbolic-refinement-correct}.
    Here, it suffices to show that the property holds before the loop --- since $\sem{\bigwedge R \wedge \bigwedge T}_\mathcal{L}$ never grows inside the loop body, preservation is easy.
    Thus, suppose that $c_1 \sem{\bigwedge T}_\mathcal{L} c_2$ where $T$ is initialized to $I$ as given above, and also that $\floor{c_1} \mathrel{\reach_\phi} \floor{c_2}$.
    Assume towards a contradiction that $c_1 \in F$ if and only if $c_2 \not\in F$.
    In that case, $\floor{c_1} = t_\acc$ if and only if $\floor{c_2} \neq t_\acc$.
    Then, $\floor{c_1}^< \wedge \floor{c_2}^> \implies \bot \in T$, and thus $c_1 \sem{\bigwedge T}_\mathcal{L} c_2$ does \emph{not} hold --- a contradiction.
    Our assumption must have been wrong, and therefore $c_1 \in F$ iff $c_2 \in F$.

    \item
    If $\rho$ is a symbolic bisimulation, then $\rho \vDash \bigwedge R \wedge \bigwedge T$.
    In this case, it again suffices to show that this property holds before the loop.
    To this end, let $c_1 \sem{\rho}_\mathcal{L} c_2$.
    For all $t_1 \mathrel{\reach_\phi} t_2$ with $t_1 = t_\acc \iff t_2 \neq t_\acc$, we should argue that $c_1 \sem{t_1^< \wedge t_2^>}_\mathcal{L} c_2$ does \emph{not} hold.
    Thus, suppose towards a contradiction that $c_1 \sem{t_1^< \wedge t_2^>}_\mathcal{L} c_2$.
    Now, if $c_1 \in F$, then $t_1 = t_\acc$; but then $t_2 \neq t_\acc$, meaning that $c_2 \not\in F$.
    This contradicts that $\rho$ is a symbolic bisimulation with $c_1 \sem{\rho}_\mathcal{L} c_2$.
    Thus, $c_1 \sem{t_1^< \wedge t_2^>}_\mathcal{L} c_2$ does \emph{not} hold --- we are done.
\end{enumerate}
When the loop terminates, all three invariants still hold.
Specifically, (1)~$\sem{\bigwedge R}_\mathcal{L}$ is preserved by $\delta$, (2) if $c_1 \sem{\bigwedge R}_\mathcal{L} c_2$ and $\floor{c_1} \mathrel{\reach_\phi} \floor{c_2}$, then $c_1 \in F$ if and only if $c_2 \in F$, and (3)~if $\rho$ is a symbolic bisimulation, then $\rho \vDash \bigwedge R$.
\begin{itemize}
    \item
    Suppose the algorithm returns $\True$.
    In that case, $\phi \vDash \bigwedge R$.
    We choose $\psi = \bigvee\nolimits_{t_1 \mathrel{\reach_\phi} t_2} (t_1^< \wedge t_2^>)$, and claim that $\psi \wedge \bigwedge R$ is a symbolic bisimulation.
    Clearly, both $\sem{\psi}_\mathcal{L}$ and $\sem{\bigwedge R}_\mathcal{L}$ preserve $\delta$ --- the former by construction, the latter by the first loop invariant.

    To see that $\psi \wedge \bigwedge R$ is compatible with $F$, suppose $c_1 \sem{\psi \wedge \bigwedge R}_\mathcal{L} c_2$.
    In that case, we have $c_1 \sem{\psi}_\mathcal{L} c_2$, and so $\floor{c_2} \mathrel{\reach_\phi} \floor{c_2}$.
    Also, since $c_1 \sem{\bigwedge R}_\mathcal{L} c_2$, we know that $c_1 \in F$ if and only if $c_2 \in F$ by the second loop invariant.
    Thus, $\psi \wedge \bigwedge R$ is indeed a symbolic bisimulation.
    Since $\phi \vDash \psi \wedge \bigwedge R$, we can conclude that $c_1 \sem{\phi}_\mathcal{L} c_2$ implies $L(c_1) = L(c_2)$ by \cref{lemma:symbolic-bisimulation-sound-complete}.

    \item
    Suppose that $c_1 \sem{\phi}_\mathcal{L} c_2$ implies $L(c_1) = L(c_2)$.
    By \cref{lemma:symbolic-bisimulation-sound-complete} there exists some symbolic bisimulation $\psi$ such that $\phi \vDash \psi$.
    By the loop invariant, we know that $\psi \vDash \bigwedge R$.
    It then follows that $\phi \vDash \bigwedge R$, and thus the algorithm returns $\True$.
    \qedhere
\end{itemize}
\end{proof}

\restatebisimulationwithleapssoundcorrect*
\begin{proof}
The backward implication is straightforward.
In this case, we note that $\phi \vDash \psi$ for some symbolic bisimulation $\psi$.
Because any (symbolic) bisimulation is in particular a (symbolic) bisimulation with leaps, the claim follows.

For the other direction, suppose that $\psi$ is a symbolic bisimulation with leaps such that $\phi \vDash \psi$.
We define $R$ as the smallest relation satisfying
\[
    \inferrule{%
        c_1 \mathrel{\sem{\psi}_\mathcal{L}} c_2 \\
        w \in \bits^*
    }{%
        \delta^*(c_1, w) \mathrel{R} \delta^*(c_2, w)
    }
\]
Clearly, $R$ is closed under steps by construction.

We claim that $R$ is a bisimulation.
It suffices to prove that for all $c_1 \sem{\psi}_\mathcal{L} c_2$ and $w \in \bits$, we have $\delta^*(c_1, w) \in F$ if and only if $\delta^*(c_2, w) \in F$.
We proceed by induction on $|w|$.

In the base, $w = \epsilon$.
We can then conclude that $\delta^*(c_1, w) = c_1 \in F$ if and only if $\delta^*(c_2, w) = c_2 \in F$, because $c_1 \sem{\psi}_\mathcal{L} c_2$ and $\psi$ is a symbolic bisimulation with leaps.

For the inductive step, let $|w| > 0$ and $n = \sharp(c_1,c_2)$, and assume that the claim holds for all $y \in \bits$ with $|y| < |w|$.
On the one hand, if $|w| < n$, then necessarily $n > 1$, and thus $c_1, c_2 \not\in F$.
This means that $\delta^*(c_1, w), \delta^*(c_2, w) \not\in F$, because the state component of those configurations does not change in the first $n$ steps.
On the other hand, if $|w| \geq n$, then we write $w = xy$ with $|x| = n$.
Now, since $\sem{\psi}_\mathcal{L}$ is a symbolic bisimulation with leaps we know that $\delta^*(c_1, x) \sem{\psi}_\mathcal{L} \delta^*(c_2, x)$.
Finally, since $|y| < |w|$, we have
\begin{mathpar}
  \delta^*(c_1, w) = \delta^*(\delta^*(c_1, x), y) \in F \iff \\
  \delta^*(c_2, w) = \delta^*(\delta^*(c_2, x), y) \in F
\end{mathpar}
by induction.
\end{proof}

\begin{figure*}
\centering
\begin{tabular}{ccc}
\begin{subfigure}{0.3\textwidth}
\begin{lf-grey}
parse_0 {
  extract(T0, 8);
  extract(L0, 8);
  select(T0, L0) {
    (0x00, 0x00) => goto accept
    (0x01, 0x00) => goto accept
    (_, 0x01) => goto parse_v01
    (_, 0x02) => goto parse_v02
    (_, 0x03) => goto parse_v03
    (_, 0x04) => goto parse_v04
    (_, 0x05) => goto parse_v05
    (_, 0x06) => goto parse_v06
  }
}
parse_v01 {
  extract(scratch, 8);
  v0 <- scratch ++ v0[7:47]
  goto parse_1
}
parse_v02 {
  extract(scratch, 16);
  v0 <- scratch ++ v0[15:47]
  goto parse_1
}
parse_v03 {
  extract(scratch, 24);
  v0 <- scratch ++ v0[23:47]
  goto parse_1
}
parse_v04 {
  extract(scratch, 32);
  v0 <- scratch ++ v0[31:47]
  goto parse_1
}
parse_v05 {
  extract(scratch, 40);
  v0 <- scratch ++ v0[39:47]
  goto parse_1
}
parse_v06 {
  extract(v0, 48);
  goto parse_1
}
\end{lf-grey}
\end{subfigure}
&
\begin{subfigure}{0.3\textwidth}
\begin{lf-grey}
parse_1 {
  extract(T1, 8);
  extract(L1, 8);
  select(T1, L1) {
    (0x00, 0x00) => goto accept
    (0x01, 0x00) => goto accept
    (_, 0x01) => goto parse_v11
    (_, 0x02) => goto parse_v12
    (_, 0x03) => goto parse_v13
    (_, 0x04) => goto parse_v14
    (_, 0x05) => goto parse_v15
    (_, 0x06) => goto parse_v16
  }
}
parse_v11 {
  extract(scratch, 8);
  v1 <- scratch ++ v1[7:47]
  goto parse_2
}
parse_v12 {
  extract(scratch, 16);
  v1 <- scratch ++ v1[15:47]
  goto parse_2
}
parse_v13 {
  extract(scratch, 24);
  v1 <- scratch ++ v1[23:47]
  goto parse_2
}
parse_v14 {
  extract(scratch, 32);
  v1 <- scratch ++ v1[31:47]
  goto parse_2
}
parse_v15 {
  extract(scratch, 40);
  v1 <- scratch ++ v1[39:47]
  goto parse_2
}
parse_v16 {
  extract(v1, 48);
  goto parse_2
}
\end{lf-grey}
\end{subfigure}
&
\begin{subfigure}{0.3\textwidth}
\begin{lf-grey}
parse_2 {
  extract(T2, 8);
  extract(L2, 8);
  select(T2, L2) {
    (0x00, 0x00) => goto accept
    (0x01, 0x00) => goto accept
    (_, 0x01) => goto parse_v21
    (_, 0x02) => goto parse_v22
    (_, 0x03) => goto parse_v23
    (_, 0x04) => goto parse_v24
    (_, 0x05) => goto parse_v25
    (_, 0x06) => goto parse_v26
  }
}
parse_v21 {
  extract(scratch, 8);
  v2 <- scratch ++ v2[7:47]
  goto accept
}
parse_v22 {
  extract(scratch, 16);
  v2 <- scratch ++ v2[15:47]
  goto accept
}
parse_v23 {
  extract(scratch, 24);
  v2 <- scratch ++ v2[23:47]
  goto accept
}
parse_v24 {
  extract(scratch, 32);
  v2 <- scratch ++ v2[31:47]
  goto accept
}
parse_v25 {
  extract(scratch, 40);
  v2 <- scratch ++ v2[39:47]
  goto accept
}
parse_v26 {
  extract(v2, 48);
  goto accept
}
\end{lf-grey}
\end{subfigure}
\end{tabular}
\caption{Generic IP options parser.}
\label{fig:self-comp3-1}
\end{figure*}

\begin{figure*}
\centering
\begin{tabular}{ccc}
\begin{subfigure}{0.3\textwidth}
\begin{lf-grey}
parse_0 {
  extract(T0, 8);
  extract(L0, 8);
  select(T0, L0) {
    (0x00, 0x00) => goto accept
    (0x01, 0x00) => goto accept
    (0x44, 0x06) => goto parse_stamp0
    (_, 0x01) => goto parse_v01
    (_, 0x02) => goto parse_v02
    (_, 0x03) => goto parse_v03
    (_, 0x04) => goto parse_v04
    (_, 0x05) => goto parse_v05
    (_, 0x06) => goto parse_v06
  }
}
parse_stamp0 {
  extract(ptr0, 8);
  extract(over0, 4);
  extract(flag0, 4);
  extract(time0, 32);
  goto parse_1
}
parse_v01 {
  extract(scratch, 8);
  v0 <- scratch ++ v0[7:47]
  goto parse_1
}
parse_v02 {
  extract(scratch, 16);
  v0 <- scratch ++ v0[15:47]
  goto parse_1
}
parse_v03 {
  extract(scratch, 24);
  v0 <- scratch ++ v0[23:47]
  goto parse_1
}
parse_v04 {
  extract(scratch, 32);
  v0 <- scratch ++ v0[31:47]
  goto parse_1
}
parse_v05 {
  extract(scratch, 40);
  v0 <- scratch ++ v0[39:47]
  goto parse_1
}
parse_v06 {
  extract(v0, 48);
  goto parse_1
}
\end{lf-grey}
\end{subfigure}
&
\begin{subfigure}{0.3\textwidth}
\begin{lf-grey}
parse_1 {
  extract(T1, 8);
  extract(L1, 8);
  select(T1, L1) {
    (0x00, 0x00) => goto accept
    (0x01, 0x00) => goto accept
    (0x44, 0x06) => goto parse_stamp1
    (_, 0x01) => goto parse_v11
    (_, 0x02) => goto parse_v12
    (_, 0x03) => goto parse_v13
    (_, 0x04) => goto parse_v14
    (_, 0x05) => goto parse_v15
    (_, 0x06) => goto parse_v16
  }
}
parse_stamp1 {
  extract(ptr1, 8);
  extract(over1, 4);
  extract(flag1, 4);
  extract(time1, 32);
  goto parse_2
}
parse_v11 {
  extract(scratch, 8);
  v1 <- scratch ++ v1[7:47]
  goto parse_2
}
parse_v12 {
  extract(scratch, 16);
  v1 <- scratch ++ v1[15:47]
  goto parse_2
}
parse_v13 {
  extract(scratch, 24);
  v1 <- scratch ++ v1[23:47]
  goto parse_2
}
parse_v14 {
  extract(scratch, 32);
  v1 <- scratch ++ v1[31:47]
  goto parse_2
}
parse_v15 {
  extract(scratch, 40);
  v1 <- scratch ++ v1[39:47]
  goto parse_2
}
parse_v16 {
  extract(v1, 48);
  goto parse_2
}
\end{lf-grey}
\end{subfigure}&
\begin{subfigure}{0.3\textwidth}
\begin{lf-grey}
parse_2 {
  extract(T2, 8);
  extract(L2, 8);
  select(T2, L2) {
    (0x00, 0x00) => goto accept
    (0x01, 0x00) => goto accept
    (0x44, 0x06) => goto parse_stamp2
    (_, 0x01) => goto parse_v21
    (_, 0x02) => goto parse_v22
    (_, 0x03) => goto parse_v23
    (_, 0x04) => goto parse_v24
    (_, 0x05) => goto parse_v25
    (_, 0x06) => goto parse_v26
  }
}
parse_stamp2 {
  extract(ptr2, 8);
  extract(over2, 4);
  extract(flag2, 4);
  extract(time2, 32);
  goto accept
}
parse_v21 {
  extract(scratch, 8);
  v2 <- scratch ++ v2[7:47]
  goto accept
}
parse_v22 {
  extract(scratch, 16);
  v2 <- scratch ++ v2[15:47]
  goto accept
}
parse_v23 {
  extract(scratch, 24);
  v2 <- scratch ++ v2[23:47]
  goto accept
}
parse_v24 {
  extract(scratch, 32);
  v2 <- scratch ++ v2[31:47]
  goto accept
}
parse_v25 {
  extract(scratch, 40);
  v2 <- scratch ++ v2[39:47]
  goto accept
}
parse_v26 {
  extract(v2, 48);
  goto accept
}
\end{lf-grey}
\end{subfigure}
\end{tabular}
\caption{Specialized IP options Timestamp parser.}%
\label{fig:self-comp3-2}
\end{figure*}

%% file: petr4-parsers.bbl
%%% -*-BibTeX-*-
%%% Do NOT edit. File created by BibTeX with style
%%% ACM-Reference-Format-Journals [18-Jan-2012].

\begin{thebibliography}{57}

%%% ====================================================================
%%% NOTE TO THE USER: you can override these defaults by providing
%%% customized versions of any of these macros before the \bibliography
%%% command.  Each of them MUST provide its own final punctuation,
%%% except for \shownote{}, \showDOI{}, and \showURL{}.  The latter two
%%% do not use final punctuation, in order to avoid confusing it with
%%% the Web address.
%%%
%%% To suppress output of a particular field, define its macro to expand
%%% to an empty string, or better, \unskip, like this:
%%%
%%% \newcommand{\showDOI}[1]{\unskip}   % LaTeX syntax
%%%
%%% \def \showDOI #1{\unskip}           % plain TeX syntax
%%%
%%% ====================================================================

\ifx \showCODEN    \undefined \def \showCODEN     #1{\unskip}     \fi
\ifx \showDOI      \undefined \def \showDOI       #1{#1}\fi
\ifx \showISBNx    \undefined \def \showISBNx     #1{\unskip}     \fi
\ifx \showISBNxiii \undefined \def \showISBNxiii  #1{\unskip}     \fi
\ifx \showISSN     \undefined \def \showISSN      #1{\unskip}     \fi
\ifx \showLCCN     \undefined \def \showLCCN      #1{\unskip}     \fi
\ifx \shownote     \undefined \def \shownote      #1{#1}          \fi
\ifx \showarticletitle \undefined \def \showarticletitle #1{#1}   \fi
\ifx \showURL      \undefined \def \showURL       {\relax}        \fi
% The following commands are used for tagged output and should be
% invisible to TeX
\providecommand\bibfield[2]{#2}
\providecommand\bibinfo[2]{#2}
\providecommand\natexlab[1]{#1}
\providecommand\showeprint[2][]{arXiv:#2}

\bibitem[\protect\citeauthoryear{Appel}{Appel}{2011}]%
        {appel2011verified}
\bibfield{author}{\bibinfo{person}{Andrew~W. Appel}.}
  \bibinfo{year}{2011}\natexlab{}.
\newblock \showarticletitle{Verified Software Toolchain}. In
  \bibinfo{booktitle}{\emph{Proc. of European Symposium on Programming}}
  \emph{(\bibinfo{series}{ESOP})}. \bibinfo{pages}{1--17}.
\newblock
\urldef\tempurl%
\url{https://doi.org/10.1007/978-3-642-19718-5\_1}
\showDOI{\tempurl}


\bibitem[\protect\citeauthoryear{Appel, Beringer, Chlipala, Pierce, Shao,
  Weirich, and Zdancewic}{Appel et~al\mbox{.}}{2017}]%
        {deepspec}
\bibfield{author}{\bibinfo{person}{Andrew~W. Appel}, \bibinfo{person}{Lennart
  Beringer}, \bibinfo{person}{Adam Chlipala}, \bibinfo{person}{Benjamin~C.
  Pierce}, \bibinfo{person}{Zhong Shao}, \bibinfo{person}{Stephanie Weirich},
  {and} \bibinfo{person}{Steve Zdancewic}.} \bibinfo{year}{2017}\natexlab{}.
\newblock \showarticletitle{Position paper: the Science of Deep Specification}.
\newblock \bibinfo{journal}{\emph{Philosophical Transactions of the Royal
  Society A: Mathematical, Physical and Engineering Sciences}}
  \bibinfo{volume}{375}, \bibinfo{number}{2104} (\bibinfo{year}{2017}),
  \bibinfo{pages}{20160331}.
\newblock
\urldef\tempurl%
\url{https://doi.org/10.1098/rsta.2016.0331}
\showDOI{\tempurl}


\bibitem[\protect\citeauthoryear{Appel, Dockins, Hobor, Beringer, Dodds,
  Stewart, Blazy, and Leroy}{Appel et~al\mbox{.}}{2014}]%
        {vstbook}
\bibfield{author}{\bibinfo{person}{Andrew~W. Appel}, \bibinfo{person}{Robert
  Dockins}, \bibinfo{person}{Aquinas Hobor}, \bibinfo{person}{Lennart
  Beringer}, \bibinfo{person}{Josiah Dodds}, \bibinfo{person}{Gordon Stewart},
  \bibinfo{person}{Sandrine Blazy}, {and} \bibinfo{person}{Xavier Leroy}.}
  \bibinfo{year}{2014}\natexlab{}.
\newblock \bibinfo{booktitle}{\emph{Program Logics for Certified Compilers}}.
\newblock \bibinfo{publisher}{Cambridge University Press}.
\newblock
\showISBNx{978-1-10-704801-0}


\bibitem[\protect\citeauthoryear{Armand, Faure, Gr{\'{e}}goire, Keller,
  Th{\'{e}}ry, and Werner}{Armand et~al\mbox{.}}{2011}]%
        {armand-etal-2011}
\bibfield{author}{\bibinfo{person}{Micha{\"{e}}l Armand},
  \bibinfo{person}{Germain Faure}, \bibinfo{person}{Benjamin Gr{\'{e}}goire},
  \bibinfo{person}{Chantal Keller}, \bibinfo{person}{Laurent Th{\'{e}}ry},
  {and} \bibinfo{person}{Benjamin Werner}.} \bibinfo{year}{2011}\natexlab{}.
\newblock \showarticletitle{A Modular Integration of {SAT/SMT} Solvers to Coq
  through Proof Witnesses}. In \bibinfo{booktitle}{\emph{Proc. of the
  International Conference on Certified Programs and Proofs}}
  \emph{(\bibinfo{series}{CPP})}. \bibinfo{pages}{135--150}.
\newblock
\urldef\tempurl%
\url{https://doi.org/10.1007/978-3-642-25379-9\_12}
\showDOI{\tempurl}


\bibitem[\protect\citeauthoryear{Authority}{Authority}{2018}]%
        {ip-params}
\bibfield{author}{\bibinfo{person}{Internet Assigned~Numbers Authority}.}
  \bibinfo{year}{2018}\natexlab{}.
\newblock \bibinfo{booktitle}{\emph{Internet Protocol Version 4 (IPv4)
  Parameters}}.
\newblock
\urldef\tempurl%
\url{https://www.iana.org/assignments/ip-parameters/ip-parameters.xhtml}
\showURL{%
\tempurl}


\bibitem[\protect\citeauthoryear{Barrett, Conway, Deters, Hadarean, Jovanovic,
  King, Reynolds, and Tinelli}{Barrett et~al\mbox{.}}{2011}]%
        {barrett-etal-2011}
\bibfield{author}{\bibinfo{person}{Clark~W. Barrett},
  \bibinfo{person}{Christopher~L. Conway}, \bibinfo{person}{Morgan Deters},
  \bibinfo{person}{Liana Hadarean}, \bibinfo{person}{Dejan Jovanovic},
  \bibinfo{person}{Tim King}, \bibinfo{person}{Andrew Reynolds}, {and}
  \bibinfo{person}{Cesare Tinelli}.} \bibinfo{year}{2011}\natexlab{}.
\newblock \showarticletitle{{CVC4}}. In \bibinfo{booktitle}{\emph{Proc. of
  Computer Aided Verification}} \emph{(\bibinfo{series}{CAV})}.
  \bibinfo{pages}{171--177}.
\newblock
\urldef\tempurl%
\url{https://doi.org/10.1007/978-3-642-22110-1\_14}
\showDOI{\tempurl}


\bibitem[\protect\citeauthoryear{Barrett, Stump, and Tinelli}{Barrett
  et~al\mbox{.}}{2010}]%
        {barrett-etal-2010}
\bibfield{author}{\bibinfo{person}{Clark~W. Barrett}, \bibinfo{person}{Aaron
  Stump}, {and} \bibinfo{person}{Cesare Tinelli}.}
  \bibinfo{year}{2010}\natexlab{}.
\newblock \showarticletitle{{The SMT-LIB Standard: Version 2.0}}. In
  \bibinfo{booktitle}{\emph{Proc. of the 8th International Workshop on
  Satisfiability Modulo Theories}} \emph{(\bibinfo{series}{SMT},
  Vol.~\bibinfo{volume}{13})}. \bibinfo{pages}{14--14}.
\newblock


\bibitem[\protect\citeauthoryear{Bertot and Cast{\'{e}}ran}{Bertot and
  Cast{\'{e}}ran}{2004}]%
        {bertot2013interactive}
\bibfield{author}{\bibinfo{person}{Yves Bertot} {and} \bibinfo{person}{Pierre
  Cast{\'{e}}ran}.} \bibinfo{year}{2004}\natexlab{}.
\newblock \bibinfo{booktitle}{\emph{Interactive Theorem Proving and Program
  Development - Coq'Art: The Calculus of Inductive Constructions}}.
\newblock
\showISBNx{978-3-642-05880-6}
\urldef\tempurl%
\url{https://doi.org/10.1007/978-3-662-07964-5}
\showDOI{\tempurl}


\bibitem[\protect\citeauthoryear{Bonchi and Pous}{Bonchi and Pous}{2013}]%
        {bonchi-pous-2013}
\bibfield{author}{\bibinfo{person}{Filippo Bonchi} {and}
  \bibinfo{person}{Damien Pous}.} \bibinfo{year}{2013}\natexlab{}.
\newblock \showarticletitle{Checking NFA Equivalence with Bisimulations up to
  Congruence}. In \bibinfo{booktitle}{\emph{Proc. of the 40th Annual ACM
  SIGPLAN-SIGACT Symposium on Principles of Programming Languages}}
  \emph{(\bibinfo{series}{POPL})}. \bibinfo{pages}{457–468}.
\newblock
\showISBNx{9781450318327}
\urldef\tempurl%
\url{https://doi.org/10.1145/2429069.2429124}
\showDOI{\tempurl}


\bibitem[\protect\citeauthoryear{Bosshart, Daly, Gibb, Izzard, McKeown,
  Rexford, Schlesinger, Talayco, Vahdat, Varghese, and Walker}{Bosshart
  et~al\mbox{.}}{2014}]%
        {p4paper}
\bibfield{author}{\bibinfo{person}{Pat Bosshart}, \bibinfo{person}{Dan Daly},
  \bibinfo{person}{Glen Gibb}, \bibinfo{person}{Martin Izzard},
  \bibinfo{person}{Nick McKeown}, \bibinfo{person}{Jennifer Rexford},
  \bibinfo{person}{Cole Schlesinger}, \bibinfo{person}{Dan Talayco},
  \bibinfo{person}{Amin Vahdat}, \bibinfo{person}{George Varghese}, {and}
  \bibinfo{person}{David Walker}.} \bibinfo{year}{2014}\natexlab{}.
\newblock \showarticletitle{P4: Programming Protocol-Independent Packet
  Processors}.
\newblock \bibinfo{journal}{\emph{ACM SIGCOMM Comput. Commun. Rev.}}
  \bibinfo{volume}{44}, \bibinfo{number}{3} (\bibinfo{date}{July}
  \bibinfo{year}{2014}), \bibinfo{pages}{87–95}.
\newblock
\showISSN{0146-4833}
\urldef\tempurl%
\url{https://doi.org/10.1145/2656877.2656890}
\showDOI{\tempurl}


\bibitem[\protect\citeauthoryear{Bouajjani, Fernandez, and Halbwachs}{Bouajjani
  et~al\mbox{.}}{1991}]%
        {bouajjani-etal-1990}
\bibfield{author}{\bibinfo{person}{Ahmed Bouajjani},
  \bibinfo{person}{Jean{-}Claude Fernandez}, {and} \bibinfo{person}{Nicolas
  Halbwachs}.} \bibinfo{year}{1991}\natexlab{}.
\newblock \showarticletitle{Minimal Model Generation}. In
  \bibinfo{booktitle}{\emph{Proc. of the 2nd International Workshop on Computer
  Aided Verification}} \emph{(\bibinfo{series}{CAV})}.
  \bibinfo{pages}{197--203}.
\newblock
\urldef\tempurl%
\url{https://doi.org/10.1007/BFb0023733}
\showDOI{\tempurl}


\bibitem[\protect\citeauthoryear{Bouajjani, Fernandez, Halbwachs, and
  Raymond}{Bouajjani et~al\mbox{.}}{1992}]%
        {bouajjani-etal-1992}
\bibfield{author}{\bibinfo{person}{Ahmed Bouajjani},
  \bibinfo{person}{Jean{-}Claude Fernandez}, \bibinfo{person}{Nicolas
  Halbwachs}, {and} \bibinfo{person}{Pascal Raymond}.}
  \bibinfo{year}{1992}\natexlab{}.
\newblock \showarticletitle{Minimal State Graph Generation}.
\newblock \bibinfo{journal}{\emph{Science of Computer Programming}}
  \bibinfo{volume}{18}, \bibinfo{number}{3} (\bibinfo{year}{1992}),
  \bibinfo{pages}{247--269}.
\newblock
\urldef\tempurl%
\url{https://doi.org/10.1016/0167-6423(92)90018-7}
\showDOI{\tempurl}


\bibitem[\protect\citeauthoryear{Bouali and de~Simone}{Bouali and
  de~Simone}{1992}]%
        {bouali-desimone-1992}
\bibfield{author}{\bibinfo{person}{Amar Bouali} {and} \bibinfo{person}{Robert
  de Simone}.} \bibinfo{year}{1992}\natexlab{}.
\newblock \showarticletitle{Symbolic Bisimulation Minimisation}. In
  \bibinfo{booktitle}{\emph{Proc. of the 4th International Workshop on Computer
  Aided Verification}} \emph{(\bibinfo{series}{CAV})}.
  \bibinfo{pages}{96--108}.
\newblock
\urldef\tempurl%
\url{https://doi.org/10.1007/3-540-56496-9\_9}
\showDOI{\tempurl}


\bibitem[\protect\citeauthoryear{Burch, Clarke, McMillan, Dill, and
  Hwang}{Burch et~al\mbox{.}}{1992}]%
        {burch-etal-1992}
\bibfield{author}{\bibinfo{person}{Jerry~R. Burch}, \bibinfo{person}{Edmund~M.
  Clarke}, \bibinfo{person}{Kenneth~L. McMillan}, \bibinfo{person}{David~L.
  Dill}, {and} \bibinfo{person}{Lucius~J. Hwang}.}
  \bibinfo{year}{1992}\natexlab{}.
\newblock \showarticletitle{Symbolic Model Checking: $10^{20}$ States and
  Beyond}.
\newblock \bibinfo{journal}{\emph{Information and Computation}}
  \bibinfo{volume}{98}, \bibinfo{number}{2} (\bibinfo{year}{1992}),
  \bibinfo{pages}{142--170}.
\newblock
\urldef\tempurl%
\url{https://doi.org/10.1016/0890-5401(92)90017-A}
\showDOI{\tempurl}


\bibitem[\protect\citeauthoryear{Churchill, Padon, Sharma, and Aiken}{Churchill
  et~al\mbox{.}}{2019}]%
        {10.1145/3314221.3314596}
\bibfield{author}{\bibinfo{person}{Berkeley Churchill}, \bibinfo{person}{Oded
  Padon}, \bibinfo{person}{Rahul Sharma}, {and} \bibinfo{person}{Alex Aiken}.}
  \bibinfo{year}{2019}\natexlab{}.
\newblock \showarticletitle{Semantic Program Alignment for Equivalence
  Checking}. In \bibinfo{booktitle}{\emph{Proc. of the 40th ACM SIGPLAN
  Conference on Programming Language Design and Implementation}}
  \emph{(\bibinfo{series}{PLDI})}. \bibinfo{pages}{1027–1040}.
\newblock
\showISBNx{9781450367127}
\urldef\tempurl%
\url{https://doi.org/10.1145/3314221.3314596}
\showDOI{\tempurl}


\bibitem[\protect\citeauthoryear{Consortium}{Consortium}{2021}]%
        {p4-spec}
\bibfield{author}{\bibinfo{person}{The P4~Language Consortium}.}
  \bibinfo{year}{2021}\natexlab{}.
\newblock \bibinfo{booktitle}{\emph{P4 Language Specification, Version 1.2.2}}.
\newblock
\newblock
\shownote{Available at \url{https://p4.org/p4-spec/docs/P4-16-v1.2.2.html}.}


\bibitem[\protect\citeauthoryear{{Coq} {Development}~{Team}}{{Coq}
  {Development}~{Team}}{2021}]%
        {coq-web}
\bibfield{author}{\bibinfo{person}{The {Coq} {Development}~{Team}}.}
  \bibinfo{year}{2021}\natexlab{}.
\newblock \bibinfo{booktitle}{\emph{The {Coq} Reference Manual, version 8.14}}.
\newblock
\newblock
\shownote{Available electronically at \url{http://coq.inria.fr/doc}.}


\bibitem[\protect\citeauthoryear{Coudert, Berthet, and Madre}{Coudert
  et~al\mbox{.}}{1989}]%
        {coudert-etal-1989}
\bibfield{author}{\bibinfo{person}{Olivier Coudert}, \bibinfo{person}{Christian
  Berthet}, {and} \bibinfo{person}{Jean~Christophe Madre}.}
  \bibinfo{year}{1989}\natexlab{}.
\newblock \showarticletitle{Verification of Synchronous Sequential Machines
  Based on Symbolic Execution}. In \bibinfo{booktitle}{\emph{Proc. of Automatic
  Verification Methods for Finite State Systems}}. \bibinfo{pages}{365--373}.
\newblock
\urldef\tempurl%
\url{https://doi.org/10.1007/3-540-52148-8_30}
\showDOI{\tempurl}


\bibitem[\protect\citeauthoryear{Czajka and Kaliszyk}{Czajka and
  Kaliszyk}{2018}]%
        {czajka-kaliszyk-2018}
\bibfield{author}{\bibinfo{person}{{\L}ukasz Czajka} {and}
  \bibinfo{person}{Cezary Kaliszyk}.} \bibinfo{year}{2018}\natexlab{}.
\newblock \showarticletitle{Hammer for Coq: Automation for Dependent Type
  Theory}.
\newblock \bibinfo{journal}{\emph{Journal of Automated Reasoning}}
  \bibinfo{volume}{61}, \bibinfo{number}{1–4} (\bibinfo{date}{jun}
  \bibinfo{year}{2018}), \bibinfo{pages}{423–453}.
\newblock
\showISSN{0168-7433}
\urldef\tempurl%
\url{https://doi.org/10.1007/s10817-018-9458-4}
\showDOI{\tempurl}


\bibitem[\protect\citeauthoryear{D'Antoni, Kincaid, and Wang}{D'Antoni
  et~al\mbox{.}}{2018}]%
        {dantoni-etal-2018}
\bibfield{author}{\bibinfo{person}{Loris D'Antoni}, \bibinfo{person}{Zachary
  Kincaid}, {and} \bibinfo{person}{Fang Wang}.}
  \bibinfo{year}{2018}\natexlab{}.
\newblock \showarticletitle{A Symbolic Decision Procedure for Symbolic
  Alternating Finite Automata}. In \bibinfo{booktitle}{\emph{Proc. of the 33rd
  International Conference on the Mathematical Foundations of Programming
  Semantics (MFPS)}}. \bibinfo{pages}{79--99}.
\newblock
\urldef\tempurl%
\url{https://doi.org/10.1016/j.entcs.2018.03.017}
\showDOI{\tempurl}


\bibitem[\protect\citeauthoryear{de~Moura and Bj{\o}rner}{de~Moura and
  Bj{\o}rner}{2008}]%
        {de2008z3}
\bibfield{author}{\bibinfo{person}{Leonardo de Moura} {and}
  \bibinfo{person}{Nikolaj Bj{\o}rner}.} \bibinfo{year}{2008}\natexlab{}.
\newblock \showarticletitle{{Z3:} An Efficient {SMT} Solver}. In
  \bibinfo{booktitle}{\emph{Proc. of the 14th International Conference on Tools
  and Algorithms for the Construction and Analysis of Systems}}
  \emph{(\bibinfo{series}{TACAS})}. \bibinfo{pages}{337--340}.
\newblock
\urldef\tempurl%
\url{https://doi.org/10.1007/978-3-540-78800-3\_24}
\showDOI{\tempurl}


\bibitem[\protect\citeauthoryear{Dehnert, Katoen, and Parker}{Dehnert
  et~al\mbox{.}}{2013}]%
        {dehnert-etal-2013}
\bibfield{author}{\bibinfo{person}{Christian Dehnert},
  \bibinfo{person}{Joost{-}Pieter Katoen}, {and} \bibinfo{person}{David
  Parker}.} \bibinfo{year}{2013}\natexlab{}.
\newblock \showarticletitle{SMT-Based Bisimulation Minimisation of Markov
  Models}. In \bibinfo{booktitle}{\emph{Proc. of the 14th International
  Workshop on Verification, Model Checking, and Abstract Interpretation}}
  \emph{(\bibinfo{series}{VMCAI})}. \bibinfo{pages}{28--47}.
\newblock
\urldef\tempurl%
\url{https://doi.org/10.1007/978-3-642-35873-9\_5}
\showDOI{\tempurl}


\bibitem[\protect\citeauthoryear{Delaware, Suriyakarn, Pit{-}Claudel, Ye, and
  Chlipala}{Delaware et~al\mbox{.}}{2019}]%
        {narcissus}
\bibfield{author}{\bibinfo{person}{Benjamin Delaware}, \bibinfo{person}{Sorawit
  Suriyakarn}, \bibinfo{person}{Cl{\'{e}}ment Pit{-}Claudel},
  \bibinfo{person}{Qianchuan Ye}, {and} \bibinfo{person}{Adam Chlipala}.}
  \bibinfo{year}{2019}\natexlab{}.
\newblock \showarticletitle{Narcissus: Correct-by-Construction Derivation of
  Decoders and Encoders from Binary Formats}.
\newblock \bibinfo{journal}{\emph{Proceedings of the ACM on Programming
  Languages}} \bibinfo{volume}{3}, \bibinfo{number}{ICFP}, Article
  \bibinfo{articleno}{82} (\bibinfo{date}{July} \bibinfo{year}{2019}),
  \bibinfo{numpages}{29}~pages.
\newblock
\urldef\tempurl%
\url{https://doi.org/10.1145/3341686}
\showDOI{\tempurl}


\bibitem[\protect\citeauthoryear{Derisavi}{Derisavi}{2007}]%
        {derisavi-2007}
\bibfield{author}{\bibinfo{person}{Salem Derisavi}.}
  \bibinfo{year}{2007}\natexlab{}.
\newblock \showarticletitle{A Symbolic Algorithm for Optimal Markov Chain
  Lumping}. In \bibinfo{booktitle}{\emph{Proc. of the 13th International
  Conference on Tools and Algorithms for the Construction and Analysis of
  Systems}} \emph{(\bibinfo{series}{TACAS})}. \bibinfo{pages}{139--154}.
\newblock
\urldef\tempurl%
\url{https://doi.org/10.1007/978-3-540-71209-1\_13}
\showDOI{\tempurl}


\bibitem[\protect\citeauthoryear{Dijkstra}{Dijkstra}{1975}]%
        {dijkstra-1975}
\bibfield{author}{\bibinfo{person}{Edsger~W. Dijkstra}.}
  \bibinfo{year}{1975}\natexlab{}.
\newblock \showarticletitle{Guarded Commands, Nondeterminacy and Formal
  Derivation of Programs}.
\newblock \bibinfo{journal}{\emph{Commun. ACM}} \bibinfo{volume}{18},
  \bibinfo{number}{8} (\bibinfo{date}{Aug.} \bibinfo{year}{1975}),
  \bibinfo{pages}{453–457}.
\newblock
\showISSN{0001-0782}
\urldef\tempurl%
\url{https://doi.org/10.1145/360933.360975}
\showDOI{\tempurl}


\bibitem[\protect\citeauthoryear{Feng, Deng, and Ying}{Feng
  et~al\mbox{.}}{2014}]%
        {feng-etal-2013}
\bibfield{author}{\bibinfo{person}{Yuan Feng}, \bibinfo{person}{Yuxin Deng},
  {and} \bibinfo{person}{Mingsheng Ying}.} \bibinfo{year}{2014}\natexlab{}.
\newblock \showarticletitle{Symbolic Bisimulation for Quantum Processes}.
\newblock \bibinfo{journal}{\emph{ACM Transactions on Computational Logic}}
  \bibinfo{volume}{15}, \bibinfo{number}{2}, Article \bibinfo{articleno}{14}
  (\bibinfo{date}{May} \bibinfo{year}{2014}), \bibinfo{numpages}{32}~pages.
\newblock
\showISSN{1529-3785}
\urldef\tempurl%
\url{https://doi.org/10.1145/2579818}
\showDOI{\tempurl}


\bibitem[\protect\citeauthoryear{Gibb, Varghese, Horowitz, and McKeown}{Gibb
  et~al\mbox{.}}{2013}]%
        {gibb2013design}
\bibfield{author}{\bibinfo{person}{Glen Gibb}, \bibinfo{person}{George
  Varghese}, \bibinfo{person}{Mark Horowitz}, {and} \bibinfo{person}{Nick
  McKeown}.} \bibinfo{year}{2013}\natexlab{}.
\newblock \showarticletitle{Design Principles for Packet Parsers}. In
  \bibinfo{booktitle}{\emph{Proc. of the 9th ACM/IEEE Symposium on Architecture
  for Networking and Communications Systems}} \emph{(\bibinfo{series}{ANCS})}.
  \bibinfo{pages}{13--24}.
\newblock
\urldef\tempurl%
\url{https://doi.org/10.1109/ANCS.2013.6665172}
\showDOI{\tempurl}


\bibitem[\protect\citeauthoryear{Gim\'{e}nez}{Gim\'{e}nez}{1994}]%
        {gimenez1994codifying}
\bibfield{author}{\bibinfo{person}{Eduardo Gim\'{e}nez}.}
  \bibinfo{year}{1994}\natexlab{}.
\newblock \showarticletitle{Codifying Guarded Definitions with Recursive
  Schemes}. In \bibinfo{booktitle}{\emph{Proc. of the International Workshop on
  Types for Proofs and Programs}} \emph{(\bibinfo{series}{TYPES})}.
  \bibinfo{pages}{39–59}.
\newblock
\showISBNx{3540605797}
\urldef\tempurl%
\url{https://doi.org/10.1007/3-540-60579-7_3}
\showDOI{\tempurl}


\bibitem[\protect\citeauthoryear{Gulwani, Polozov, and Singh}{Gulwani
  et~al\mbox{.}}{2017}]%
        {gulwani2017program}
\bibfield{author}{\bibinfo{person}{Sumit Gulwani}, \bibinfo{person}{Oleksandr
  Polozov}, {and} \bibinfo{person}{Rishabh Singh}.}
  \bibinfo{year}{2017}\natexlab{}.
\newblock \showarticletitle{Program Synthesis}.
\newblock \bibinfo{journal}{\emph{Foundations and Trends® in Programming
  Languages}} \bibinfo{volume}{4}, \bibinfo{number}{1-2}
  (\bibinfo{year}{2017}), \bibinfo{pages}{1--119}.
\newblock
\showISSN{2325-1107}
\urldef\tempurl%
\url{https://doi.org/10.1561/2500000010}
\showDOI{\tempurl}


\bibitem[\protect\citeauthoryear{Hennessy and Lin}{Hennessy and Lin}{1995}]%
        {hennessy-lin-1995}
\bibfield{author}{\bibinfo{person}{Matthew Hennessy} {and}
  \bibinfo{person}{Huimin Lin}.} \bibinfo{year}{1995}\natexlab{}.
\newblock \showarticletitle{Symbolic Bisimulations}.
\newblock \bibinfo{journal}{\emph{Theoretical Computer Science}}
  \bibinfo{volume}{138}, \bibinfo{number}{2} (\bibinfo{year}{1995}),
  \bibinfo{pages}{353--389}.
\newblock
\urldef\tempurl%
\url{https://doi.org/10.1016/0304-3975(94)00172-F}
\showDOI{\tempurl}


\bibitem[\protect\citeauthoryear{Hopcroft}{Hopcroft}{1971}]%
        {hopcroft-1971}
\bibfield{author}{\bibinfo{person}{John Hopcroft}.}
  \bibinfo{year}{1971}\natexlab{}.
\newblock \showarticletitle{An $n \log n$ Algorithm for Minimizing States in a
  Finite Automaton}.
\newblock In \bibinfo{booktitle}{\emph{Proceedings of an International
  Symposium on the Theory of Machines and Computations}}.
  \bibinfo{publisher}{Academic Press}, \bibinfo{pages}{189--196}.
\newblock
\urldef\tempurl%
\url{https://doi.org/10.1016/B978-0-12-417750-5.50022-1}
\showDOI{\tempurl}


\bibitem[\protect\citeauthoryear{Hopcroft and Karp}{Hopcroft and Karp}{1971}]%
        {hopcroft-karp-1971}
\bibfield{author}{\bibinfo{person}{John Hopcroft} {and}
  \bibinfo{person}{Richard~M. Karp}.} \bibinfo{year}{1971}\natexlab{}.
\newblock \bibinfo{booktitle}{\emph{A Linear Algorithm for Testing Equivalence
  of Finite Automata}}.
\newblock \bibinfo{type}{{T}echnical {R}eport} TR 71-114.
  \bibinfo{institution}{Cornell University}, \bibinfo{address}{Ithaca, NY}.
\newblock
\urldef\tempurl%
\url{https://hdl.handle.net/1813/5958}
\showURL{%
\tempurl}


\bibitem[\protect\citeauthoryear{Hur, Neis, Dreyer, and Vafeiadis}{Hur
  et~al\mbox{.}}{2013}]%
        {hur2013power}
\bibfield{author}{\bibinfo{person}{Chung-Kil Hur}, \bibinfo{person}{Georg
  Neis}, \bibinfo{person}{Derek Dreyer}, {and} \bibinfo{person}{Viktor
  Vafeiadis}.} \bibinfo{year}{2013}\natexlab{}.
\newblock \showarticletitle{The Power of Parameterization in Coinductive
  Proof}. In \bibinfo{booktitle}{\emph{Proc. of the 40th annual ACM
  SIGPLAN-SIGACT Symposium on Principles of Programming Languages}}
  \emph{(\bibinfo{series}{POPL})}. \bibinfo{pages}{193--206}.
\newblock
\urldef\tempurl%
\url{https://doi.org/10.1145/2480359.2429093}
\showDOI{\tempurl}


\bibitem[\protect\citeauthoryear{{IEEE Computer Society}}{{IEEE Computer
  Society}}{2018}]%
        {vlanieee}
\bibfield{author}{\bibinfo{person}{{IEEE Computer Society}}.}
  \bibinfo{year}{2018}\natexlab{}.
\newblock \showarticletitle{IEEE Standard for Local and Metropolitan Area
  Network--Bridges and Bridged Networks}.
\newblock \bibinfo{journal}{\emph{IEEE Std 802.1Q-2018 (Revision of IEEE Std
  802.1Q-2014)}} (\bibinfo{year}{2018}), \bibinfo{pages}{1--1993}.
\newblock
\urldef\tempurl%
\url{https://doi.org/10.1109/IEEESTD.2018.8403927}
\showDOI{\tempurl}


\bibitem[\protect\citeauthoryear{Kanellakis and Smolka}{Kanellakis and
  Smolka}{1983}]%
        {kanellakis-smolka-1983}
\bibfield{author}{\bibinfo{person}{Paris~C. Kanellakis} {and}
  \bibinfo{person}{Scott~A. Smolka}.} \bibinfo{year}{1983}\natexlab{}.
\newblock \showarticletitle{{CCS} Expressions, Finite State Processes, and
  Three Problems of Equivalence}. In \bibinfo{booktitle}{\emph{Proc. of the 2nd
  Annual ACM Symposium on Principles of Distributed Computing}}
  \emph{(\bibinfo{series}{PODC})}. \bibinfo{pages}{228–240}.
\newblock
\showISBNx{0897911105}
\urldef\tempurl%
\url{https://doi.org/10.1145/800221.806724}
\showDOI{\tempurl}


\bibitem[\protect\citeauthoryear{Kang, Kim, Song, Lee, Park, Shin, Kim, Cho,
  Choi, Hur, and Yi}{Kang et~al\mbox{.}}{2018}]%
        {crellvm}
\bibfield{author}{\bibinfo{person}{Jeehoon Kang}, \bibinfo{person}{Yoonseung
  Kim}, \bibinfo{person}{Youngju Song}, \bibinfo{person}{Juneyoung Lee},
  \bibinfo{person}{Sanghoon Park}, \bibinfo{person}{Mark~Dongyeon Shin},
  \bibinfo{person}{Yonghyun Kim}, \bibinfo{person}{Sungkeun Cho},
  \bibinfo{person}{Joonwon Choi}, \bibinfo{person}{Chung-Kil Hur}, {and}
  \bibinfo{person}{Kwangkeun Yi}.} \bibinfo{year}{2018}\natexlab{}.
\newblock \showarticletitle{Crellvm: Verified Credible Compilation for LLVM}.
  In \bibinfo{booktitle}{\emph{Proc. of the 39th ACM SIGPLAN Conference on
  Programming Language Design and Implementation}}
  \emph{(\bibinfo{series}{PLDI})}. \bibinfo{pages}{631–645}.
\newblock
\showISBNx{9781450356985}
\urldef\tempurl%
\url{https://doi.org/10.1145/3192366.3192377}
\showDOI{\tempurl}


\bibitem[\protect\citeauthoryear{Liu, Hallahan, Schlesinger, Sharif, Lee,
  Soul\'{e}, Wang, Ca\c{s}caval, McKeown, and Foster}{Liu
  et~al\mbox{.}}{2018}]%
        {p4v}
\bibfield{author}{\bibinfo{person}{Jed Liu}, \bibinfo{person}{William
  Hallahan}, \bibinfo{person}{Cole Schlesinger}, \bibinfo{person}{Milad
  Sharif}, \bibinfo{person}{Jeongkeun Lee}, \bibinfo{person}{Robert Soul\'{e}},
  \bibinfo{person}{Han Wang}, \bibinfo{person}{C\u{a}lin Ca\c{s}caval},
  \bibinfo{person}{Nick McKeown}, {and} \bibinfo{person}{Nate Foster}.}
  \bibinfo{year}{2018}\natexlab{}.
\newblock \showarticletitle{P4v: Practical Verification for Programmable Data
  Planes}. In \bibinfo{booktitle}{\emph{Proc. of the 2018 Conference of the ACM
  Special Interest Group on Data Communication}}
  \emph{(\bibinfo{series}{SIGCOMM})}. \bibinfo{pages}{490–503}.
\newblock
\showISBNx{9781450355674}
\urldef\tempurl%
\url{https://doi.org/10.1145/3230543.3230582}
\showDOI{\tempurl}


\bibitem[\protect\citeauthoryear{Mahboubi and Tassi}{Mahboubi and
  Tassi}{2021}]%
        {mahboubi2017mathematical}
\bibfield{author}{\bibinfo{person}{Assia Mahboubi} {and}
  \bibinfo{person}{Enrico Tassi}.} \bibinfo{year}{2021}\natexlab{}.
\newblock \bibinfo{booktitle}{\emph{Mathematical Components}}.
\newblock \bibinfo{publisher}{Zenodo}.
\newblock
\urldef\tempurl%
\url{https://doi.org/10.5281/zenodo.4457887}
\showDOI{\tempurl}


\bibitem[\protect\citeauthoryear{Massalin}{Massalin}{1987}]%
        {massalin1987superoptimizer}
\bibfield{author}{\bibinfo{person}{Henry Massalin}.}
  \bibinfo{year}{1987}\natexlab{}.
\newblock \showarticletitle{Superoptimizer: A Look at the Smallest Program}. In
  \bibinfo{booktitle}{\emph{Proc. of the 2nd International Conference on
  Architectual Support for Programming Languages and Operating Systems}}
  \emph{(\bibinfo{series}{ASPLOS})}. \bibinfo{pages}{122–126}.
\newblock
\showISBNx{0818608056}
\urldef\tempurl%
\url{https://doi.org/10.1145/36177.36194}
\showDOI{\tempurl}


\bibitem[\protect\citeauthoryear{Moore}{Moore}{2016}]%
        {moore-1956}
\bibfield{author}{\bibinfo{person}{Edward~F. Moore}.}
  \bibinfo{year}{2016}\natexlab{}.
\newblock \bibinfo{booktitle}{\emph{Gedanken-Experiments on Sequential
  Machines}}.
\newblock \bibinfo{publisher}{Princeton University Press},
  \bibinfo{pages}{129--154}.
\newblock
\urldef\tempurl%
\url{https://doi.org/doi:10.1515/9781400882618-006}
\showDOI{\tempurl}


\bibitem[\protect\citeauthoryear{Mumme and Ciardo}{Mumme and Ciardo}{2011}]%
        {mumme-ciardo-2011}
\bibfield{author}{\bibinfo{person}{Malcolm Mumme} {and}
  \bibinfo{person}{Gianfranco Ciardo}.} \bibinfo{year}{2011}\natexlab{}.
\newblock \showarticletitle{A Fully Symbolic Bisimulation Algorithm}. In
  \bibinfo{booktitle}{\emph{Proc. of the 5th International Workshop on
  Reachability Problems}} \emph{(\bibinfo{series}{RP})}.
  \bibinfo{pages}{218--230}.
\newblock
\urldef\tempurl%
\url{https://doi.org/10.1007/978-3-642-24288-5\_19}
\showDOI{\tempurl}


\bibitem[\protect\citeauthoryear{Necula}{Necula}{2000}]%
        {necula-tv}
\bibfield{author}{\bibinfo{person}{George~C. Necula}.}
  \bibinfo{year}{2000}\natexlab{}.
\newblock \showarticletitle{Translation Validation for an Optimizing Compiler}.
  In \bibinfo{booktitle}{\emph{Proc. of the ACM SIGPLAN 2000 Conference on
  Programming Language Design and Implementation}}
  \emph{(\bibinfo{series}{PLDI})}. \bibinfo{pages}{83–94}.
\newblock
\showISBNx{1581131992}
\urldef\tempurl%
\url{https://doi.org/10.1145/349299.349314}
\showDOI{\tempurl}


\bibitem[\protect\citeauthoryear{Neves, Freire, Schaeffer-Filho, and
  Barcellos}{Neves et~al\mbox{.}}{2018}]%
        {neves-p4}
\bibfield{author}{\bibinfo{person}{Miguel Neves}, \bibinfo{person}{Lucas
  Freire}, \bibinfo{person}{Alberto Schaeffer-Filho}, {and}
  \bibinfo{person}{Marinho Barcellos}.} \bibinfo{year}{2018}\natexlab{}.
\newblock \showarticletitle{Verification of P4 Programs in Feasible Time Using
  Assertions}. In \bibinfo{booktitle}{\emph{Proc. of the 14th International
  Conference on Emerging Networking EXperiments and Technologies}}
  \emph{(\bibinfo{series}{CoNEXT})}. \bibinfo{pages}{73–85}.
\newblock
\showISBNx{9781450360807}
\urldef\tempurl%
\url{https://doi.org/10.1145/3281411.3281421}
\showDOI{\tempurl}


\bibitem[\protect\citeauthoryear{Niemetz, Preiner, and Biere}{Niemetz
  et~al\mbox{.}}{2014}]%
        {niemetz-etal-2014}
\bibfield{author}{\bibinfo{person}{Aina Niemetz}, \bibinfo{person}{Mathias
  Preiner}, {and} \bibinfo{person}{Armin Biere}.}
  \bibinfo{year}{2014}\natexlab{}.
\newblock \showarticletitle{Boolector 2.0}.
\newblock \bibinfo{journal}{\emph{Journal on Satisfiability, Boolean Modeling
  and Computation}} \bibinfo{volume}{9}, \bibinfo{number}{1}
  (\bibinfo{year}{2014}), \bibinfo{pages}{53--58}.
\newblock
\urldef\tempurl%
\url{https://doi.org/10.3233/sat190101}
\showDOI{\tempurl}


\bibitem[\protect\citeauthoryear{Paige and Tarjan}{Paige and Tarjan}{1987}]%
        {paige-tarjan-1987}
\bibfield{author}{\bibinfo{person}{Robert Paige} {and}
  \bibinfo{person}{Robert~Endre Tarjan}.} \bibinfo{year}{1987}\natexlab{}.
\newblock \showarticletitle{Three Partition Refinement Algorithms}.
\newblock \bibinfo{journal}{\emph{SIAM J. Comput.}} \bibinfo{volume}{16},
  \bibinfo{number}{6} (\bibinfo{year}{1987}), \bibinfo{pages}{973--989}.
\newblock
\urldef\tempurl%
\url{https://doi.org/10.1137/0216062}
\showDOI{\tempurl}


\bibitem[\protect\citeauthoryear{Postel}{Postel}{1980}]%
        {rfc768}
\bibfield{author}{\bibinfo{person}{Jon Postel}.}
  \bibinfo{year}{1980}\natexlab{}.
\newblock \bibinfo{title}{{User Datagram Protocol}}.
\newblock \bibinfo{howpublished}{RFC 768}.
\newblock
\urldef\tempurl%
\url{https://doi.org/10.17487/RFC0768}
\showDOI{\tempurl}


\bibitem[\protect\citeauthoryear{Ramananandro, Delignat{-}Lavaud, Fournet,
  Swamy, Chajed, Kobeissi, and Protzenko}{Ramananandro et~al\mbox{.}}{2019}]%
        {everparse}
\bibfield{author}{\bibinfo{person}{Tahina Ramananandro},
  \bibinfo{person}{Antoine Delignat{-}Lavaud}, \bibinfo{person}{C{\'{e}}dric
  Fournet}, \bibinfo{person}{Nikhil Swamy}, \bibinfo{person}{Tej Chajed},
  \bibinfo{person}{Nadim Kobeissi}, {and} \bibinfo{person}{Jonathan
  Protzenko}.} \bibinfo{year}{2019}\natexlab{}.
\newblock \showarticletitle{{EverParse}: Verified Secure {Zero-Copy} Parsers
  for Authenticated Message Formats}. In \bibinfo{booktitle}{\emph{28th USENIX
  Security Symposium}} \emph{(\bibinfo{series}{USENIX Security})}.
  \bibinfo{pages}{1465--1482}.
\newblock
\showISBNx{978-1-939133-06-9}
\urldef\tempurl%
\url{https://www.usenix.org/conference/usenixsecurity19/presentation/delignat-lavaud}
\showURL{%
\tempurl}


\bibitem[\protect\citeauthoryear{Sassaman, Patterson, Bratus, and
  Locasto}{Sassaman et~al\mbox{.}}{2013}]%
        {langsec}
\bibfield{author}{\bibinfo{person}{Len Sassaman}, \bibinfo{person}{Meredith~L.
  Patterson}, \bibinfo{person}{Sergey Bratus}, {and}
  \bibinfo{person}{Michael~E. Locasto}.} \bibinfo{year}{2013}\natexlab{}.
\newblock \showarticletitle{Security Applications of Formal Language Theory}.
\newblock \bibinfo{journal}{\emph{{IEEE} Systems Journal}} \bibinfo{volume}{7},
  \bibinfo{number}{3} (\bibinfo{year}{2013}), \bibinfo{pages}{489--500}.
\newblock
\urldef\tempurl%
\url{https://doi.org/10.1109/JSYST.2012.2222000}
\showDOI{\tempurl}


\bibitem[\protect\citeauthoryear{Sozeau}{Sozeau}{2010}]%
        {sozeau-2010}
\bibfield{author}{\bibinfo{person}{Matthieu Sozeau}.}
  \bibinfo{year}{2010}\natexlab{}.
\newblock \showarticletitle{Equations: A Dependent Pattern-Matching Compiler}.
  In \bibinfo{booktitle}{\emph{Interactive Theorem Proving}}
  \emph{(\bibinfo{series}{ITP})}, \bibfield{editor}{\bibinfo{person}{Matt
  Kaufmann} {and} \bibinfo{person}{Lawrence~C. Paulson}} (Eds.).
  \bibinfo{pages}{419--434}.
\newblock
\showISBNx{978-3-642-14052-5}
\urldef\tempurl%
\url{https://doi.org/10.1007/978-3-642-14052-5\_29}
\showDOI{\tempurl}


\bibitem[\protect\citeauthoryear{Sozeau and Mangin}{Sozeau and Mangin}{2019}]%
        {sozeau-mangin-2019}
\bibfield{author}{\bibinfo{person}{Matthieu Sozeau} {and}
  \bibinfo{person}{Cyprien Mangin}.} \bibinfo{year}{2019}\natexlab{}.
\newblock \showarticletitle{Equations Reloaded: High-Level Dependently-Typed
  Functional Programming and Proving in Coq}.
\newblock \bibinfo{journal}{\emph{Proceedings of the ACM on Programming
  Languages}} \bibinfo{volume}{3}, \bibinfo{number}{ICFP}, Article
  \bibinfo{articleno}{86} (\bibinfo{date}{Jul} \bibinfo{year}{2019}),
  \bibinfo{numpages}{29}~pages.
\newblock
\urldef\tempurl%
\url{https://doi.org/10.1145/3341690}
\showDOI{\tempurl}


\bibitem[\protect\citeauthoryear{Streicher}{Streicher}{1993}]%
        {streicher-1993}
\bibfield{author}{\bibinfo{person}{Thomas Streicher}.}
  \bibinfo{year}{1993}\natexlab{}.
\newblock \showarticletitle{Investigations into Intensional Type Theory}.
\newblock \bibinfo{journal}{\emph{Habilitiation Thesis, Ludwig Maximilian
  Universit{\"a}t}} (\bibinfo{year}{1993}).
\newblock
\urldef\tempurl%
\url{https://www2.mathematik.tu-darmstadt.de/~streicher/HabilStreicher.pdf}
\showURL{%
\tempurl}


\bibitem[\protect\citeauthoryear{Tian, Gao, Liu, Zhai, Chen, Zhou, Dai, Yan,
  Ma, Tang, Lu, Wei, Liu, Zhang, Tian, and Yu}{Tian et~al\mbox{.}}{2021}]%
        {aquila}
\bibfield{author}{\bibinfo{person}{Bingchuan Tian}, \bibinfo{person}{Jiaqi
  Gao}, \bibinfo{person}{Mengqi Liu}, \bibinfo{person}{Ennan Zhai},
  \bibinfo{person}{Yanqing Chen}, \bibinfo{person}{Yu Zhou},
  \bibinfo{person}{Li Dai}, \bibinfo{person}{Feng Yan},
  \bibinfo{person}{Mengjing Ma}, \bibinfo{person}{Ming Tang},
  \bibinfo{person}{Jie Lu}, \bibinfo{person}{Xionglie Wei},
  \bibinfo{person}{Hongqiang~Harry Liu}, \bibinfo{person}{Ming Zhang},
  \bibinfo{person}{Chen Tian}, {and} \bibinfo{person}{Minlan Yu}.}
  \bibinfo{year}{2021}\natexlab{}.
\newblock \showarticletitle{Aquila: A Practically Usable Verification System
  for Production-Scale Programmable Data Planes}. In
  \bibinfo{booktitle}{\emph{Proc. of the 2021 ACM SIGCOMM 2021 Conference}}
  \emph{(\bibinfo{series}{SIGCOMM})}. \bibinfo{pages}{17–32}.
\newblock
\showISBNx{9781450383837}
\urldef\tempurl%
\url{https://doi.org/10.1145/3452296.3472937}
\showDOI{\tempurl}


\bibitem[\protect\citeauthoryear{Tristan and Leroy}{Tristan and Leroy}{2009}]%
        {verified-lcm}
\bibfield{author}{\bibinfo{person}{Jean{-}Baptiste Tristan} {and}
  \bibinfo{person}{Xavier Leroy}.} \bibinfo{year}{2009}\natexlab{}.
\newblock \showarticletitle{Verified Validation of Lazy Code Motion}. In
  \bibinfo{booktitle}{\emph{Proc. of the 30th ACM SIGPLAN Conference on
  Programming Language Design and Implementation}}
  \emph{(\bibinfo{series}{PLDI})}. \bibinfo{address}{New York, NY, USA},
  \bibinfo{pages}{316–326}.
\newblock
\showISBNx{9781605583921}
\urldef\tempurl%
\url{https://doi.org/10.1145/1542476.1542512}
\showDOI{\tempurl}


\bibitem[\protect\citeauthoryear{Viswanathan, Rosen, and Callon}{Viswanathan
  et~al\mbox{.}}{2001}]%
        {rfc3031}
\bibfield{author}{\bibinfo{person}{Arun Viswanathan}, \bibinfo{person}{Eric~C.
  Rosen}, {and} \bibinfo{person}{Ross Callon}.}
  \bibinfo{year}{2001}\natexlab{}.
\newblock \bibinfo{title}{{Multiprotocol Label Switching Architecture}}.
\newblock \bibinfo{howpublished}{RFC 3031}.
\newblock
\urldef\tempurl%
\url{https://doi.org/10.17487/RFC3031}
\showDOI{\tempurl}


\bibitem[\protect\citeauthoryear{Watanabe, Gopinathan, P\^{\i}rlea,
  Polikarpova, and Sergey}{Watanabe et~al\mbox{.}}{2021}]%
        {10.1145/3473589}
\bibfield{author}{\bibinfo{person}{Yasunari Watanabe}, \bibinfo{person}{Kiran
  Gopinathan}, \bibinfo{person}{George P\^{\i}rlea}, \bibinfo{person}{Nadia
  Polikarpova}, {and} \bibinfo{person}{Ilya Sergey}.}
  \bibinfo{year}{2021}\natexlab{}.
\newblock \showarticletitle{Certifying the Synthesis of Heap-Manipulating
  Programs}.
\newblock \bibinfo{journal}{\emph{Proceedings of the ACM on Programming
  Languages}} \bibinfo{volume}{5}, \bibinfo{number}{ICFP}, Article
  \bibinfo{articleno}{84} (\bibinfo{date}{Aug.} \bibinfo{year}{2021}),
  \bibinfo{numpages}{29}~pages.
\newblock
\urldef\tempurl%
\url{https://doi.org/10.1145/3473589}
\showDOI{\tempurl}


\bibitem[\protect\citeauthoryear{Weitz, Lyubomirsky, Heule, Torlak, Ernst, and
  Tatlock}{Weitz et~al\mbox{.}}{2017}]%
        {spacesearch}
\bibfield{author}{\bibinfo{person}{Konstantin Weitz}, \bibinfo{person}{Steven
  Lyubomirsky}, \bibinfo{person}{Stefan Heule}, \bibinfo{person}{Emina Torlak},
  \bibinfo{person}{Michael~D. Ernst}, {and} \bibinfo{person}{Zachary Tatlock}.}
  \bibinfo{year}{2017}\natexlab{}.
\newblock \showarticletitle{SpaceSearch: A Library for Building and Verifying
  Solver-Aided Tools}.
\newblock \bibinfo{journal}{\emph{Proceedings of the ACM on Programming
  Languages}} \bibinfo{volume}{1}, \bibinfo{number}{ICFP}, Article
  \bibinfo{articleno}{25} (\bibinfo{date}{Aug.} \bibinfo{year}{2017}),
  \bibinfo{numpages}{28}~pages.
\newblock
\urldef\tempurl%
\url{https://doi.org/10.1145/3110269}
\showDOI{\tempurl}


\bibitem[\protect\citeauthoryear{Zakowski, He, Hur, and Zdancewic}{Zakowski
  et~al\mbox{.}}{2020}]%
        {hur2020equational}
\bibfield{author}{\bibinfo{person}{Yannick Zakowski}, \bibinfo{person}{Paul
  He}, \bibinfo{person}{Chung-Kil Hur}, {and} \bibinfo{person}{Steve
  Zdancewic}.} \bibinfo{year}{2020}\natexlab{}.
\newblock \showarticletitle{An Equational Theory for Weak Bisimulation via
  Generalized Parameterized Coinduction}. In \bibinfo{booktitle}{\emph{Proc. of
  the 9th ACM SIGPLAN International Conference on Certified Programs and
  Proofs}} \emph{(\bibinfo{series}{CPP})}. \bibinfo{pages}{71–84}.
\newblock
\showISBNx{9781450370974}
\urldef\tempurl%
\url{https://doi.org/10.1145/3372885.3373813}
\showDOI{\tempurl}


\end{thebibliography}
